\documentclass{article}

\usepackage[english]{babel}
\usepackage{todonotes}
\usepackage{tikz}
\usetikzlibrary{patterns}
\usepackage{authblk}
\usepackage{pgfplots}
 \usepackage{fullpage}

\usepackage{amsthm}
\usepackage{thmtools}
\usepackage{amsmath,amssymb}
\usepackage{mathtools}
\usepackage{graphicx}
\usepackage[colorlinks=true, allcolors=blue]{hyperref}

\DeclareMathOperator*{\argmin}{arg\,min}
\newtheorem{definition}{Definition}[section]
\newtheorem{theorem}[definition]{Theorem}
\newtheorem{lemma}[definition]{Lemma}

\DeclareMathOperator{\up}{up}
\DeclareMathOperator{\ri}{right}
\DeclareMathOperator{\down}{down}
\DeclareMathOperator{\lef}{left}

\title{Constructing Optimal $L_{\infty}$ Star Discrepancy Sets}
\author[1]{Fran\c{c}ois Cl\'ement \thanks{francois.clement@lip6.fr}}
\author[1]{Carola Doerr \thanks{carola.doerr@lip6.fr}}
\author[2]{Kathrin Klamroth \thanks{klamroth@math.uni-wuppertal.de}}
  \author[3]{Lu\'is Paquete \thanks{paquete@dei.uc.pt}}

  \affil[1]{Sorbonne Universit\'e, CNRS, LIP6, Paris, France}
  \affil[2]{University of Wuppertal, School of Mathematics and Natural Sciences, IMACM, Gaußstr. 20, 42119 Wuppertal, Germany}
  \affil[3]{CISUC, Department of Informatics Engineering, University of Coimbra, Portugal}

\begin{document}
\maketitle

\begin{abstract}

The $L_{\infty}$ star discrepancy is a very well-studied measure used to quantify the uniformity of a point set distribution. Constructing optimal point sets for this measure is seen as a very hard problem in the discrepancy community. Indeed, provably optimal point sets are, up to now, known only for $n\leq 6$ in dimension 2 and $n \leq 2$ for higher dimensions.

We introduce in this paper mathematical programming formulations to construct point sets with as low $L_{\infty}$ star discrepancy as possible. Firstly, we present two models to construct optimal sets and show that there always exist optimal sets in dimension 2 with the property that no two points share a coordinate. Then, we provide possible extensions of our models to other measures, such as the extreme and periodic discrepancies. 
For the $L_{\infty}$ star discrepancy, we are able to compute optimal point sets for up to 21 points in dimension 2 and for up to 8 points in dimension 3. For $d=2$ and $n\ge 7$ points, these point sets have around a 50\% lower discrepancy than the current best point sets. Furthermore, by plotting the local discrepancy values induced by the optimal sets, we observe a clear structural difference with that of known low-discrepancy sets.

\end{abstract}

\section{Introduction}

Discrepancy measures are designed to quantify how closely a finite set of points approximates the uniform distribution in a given space. Among the different discrepancy measures, one of the most important is the $L_{\infty}$ star discrepancy. The $L_{\infty}$ star discrepancy of a finite point set $P \subseteq [0,1)^d$ measures the worst absolute difference between the Lebesgue measure of a $d$-dimensional box $[0,q)$ anchored in $(0,\ldots,0)$ and the proportion $|P \cap [0,q)|/|P|$ of points that fall inside this box. Point sets and sequences of low $L_{\infty}$ star discrepancy have been used in many different applications, such as one-shot optimization~\cite{CauwetCDLRRTTU20}, design of experiments~\cite{SantnerDoE}, computer vision~\cite{MatBuilder}, financial mathematics~\cite{GalFin} and most importantly in Quasi-Monte Carlo integration~\cite{DickP10}. In particular, the Koksma-Hlawka inequality~\cite{Hlawka,Koksma} shows that the error made in the numerical integration of a function $f$ can be bounded by a product of the total variation of $f$ and the $L_{\infty}$ star discrepancy of the points in which $f$ is evaluated.

The design of low-discrepancy sets and sequences has been extensively studied over the last 50 years (see for example the books by Chazelle~\cite{Chaz}, Matousek~\cite{Mat}, or Niederreiter~\cite{Nie92}) and has led to constructions outperforming random points such as Halton and Sobol' sequences, digital nets~\cite{DickP10}, or lattices with carefully chosen parameters~\cite{JoeKuo2008}. For $n$ points in $[0,1)^d$, the best known sets have a discrepancy of order $O(\log^{d-1}(n)/n)$~\cite{Nie92}, against $O(\sqrt{d/n})$ for random points~\cite{Doerr14lowerBoundRandomPoints}. 

\textbf{Point sets with a fixed number of points:}  
While obtaining the best asymptotic order of the $L_{\infty}$ star discrepancy of sets and sequences has been the main focus of the research community, finding the exact point sets minimizing the $L_{\infty}$ star discrepancy for given values of $n$ has received comparatively little attention.
White~\cite{whit:onop:1976} obtained optimal sets for up to $n=6$ points in dimension 2. Pillard, Cools and Vandewoestyne~\cite{PTV} proved the optimal $n=1$ sets for all dimensions for the $L_{\infty}$ star discrepancy as well as the $L_2$ and extreme discrepancies. Larcher and Pillichshammer~\cite{Larcher} then extended these results to $n=2$ points. For the $L_2$ periodic discrepancy, Hinrichs and Oettershagen~\cite{Hinrichs2014} were able to use symmetry groups to reduce the problem space and obtained optimal sets for up to $n=16$ points in dimension 2.

Even if we remove the optimality condition, obtaining good tailored point sets for specific $n$  and $d$ combinations has received relatively little attention. The main efforts in this direction are by de Rainville and Doerr~\cite{Rainville} using evolutionary algorithms, Steinerberger~\cite{StefEnerg} with an energy minimization, and by Clément, Doerr, and Paquete~\cite{CDP22,CDP23} via subset selection from good pre-existing point sets.

\textbf{Our contribution:} We present in this paper two different non-linear programming formulations to obtain optimal point sets with regards to the $L_{\infty}$ star discrepancy. While these formulations can be easily generalized to higher dimensions, modern solvers are only able to solve for up to 21 points in $d=2$ and up to 8 points in $d=3$ within a reasonable time frame on one machine with 24 cores of the MeSu cluster at Sorbonne Université. 

The obtained point sets are far better than Fibonacci sets, one of the best two-dimensional sets, with a near 50\% improvement in all the cases in 2 dimensions. Table \ref{tab:values} and Figure~\ref{fig:plotopt} highlight the clear improvement in discrepancy values brought by our optimal sets. The point set structure is also very different, as shown by Figure \ref{fig:op_im}, suggesting a possible new point set construction approach. All of our code and some figures are available at \url{https://github.com/frclement/OptiSetsDiscrepancy}.

\begin{table}[b]
    \centering
    \begin{tabular}{|c|c|c|c|c|c|c|c|c|c|c|}
    \hline
        $n$ & 1 & 2 & 3 & 4 & 5 & 6 & 7 & 8 & 9 & 10 \\
        \hline
        Optimal & $\frac{1+\sqrt{5}}{2}$ & 0.366 & 0.2847 &0.2500 & 0.2000 & 0.1667 & 0.1500 & 0.1328 & 0.1235 & 0.1111\\
        \hline
        Fibonacci & 1.0 & 0.6909 & 0.5880 & 0.4910 & 0.3528 & 0.3183 & 0.2728 & 0.2553 & 0.2270 & 0.2042 \\
        \hline\hline
        $n$ & 11 & 12 & 13 & 14 & 15 & 16 & 17 & 18 & 19 & 20 \\
        \hline
        Optimal & 0.1030 & 0.0952 & 0.0889 & 0.0837 & 0.0782 & 0.0739 & 0.06996 & 0.06667 & 0.0634 & 0.0604\\
        \hline
        Fibonacci & 0.1857 & 0.1702 & 0.1571 & 0.1459 & 0.1390 & 0.1486 & 0.1398 & 0.1320 & 0.1251 & 0.1188 \\
        \hline
    \end{tabular}
    \caption{Comparison of previously best values for low-discrepancy sets and our optimal sets}
    \label{tab:values}
\end{table}

\begin{figure}[t]
    \centering
    \includegraphics[width=0.8\textwidth]{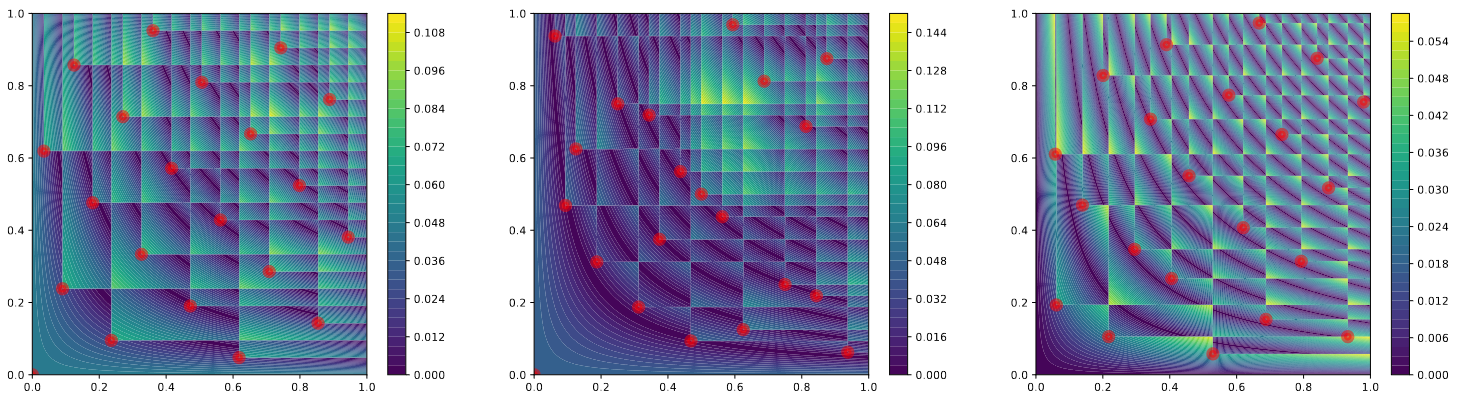}
    \caption{Fibonacci, Sobol' and optimal sets' local discrepancies for $n=21$.}
    \label{fig:op_im}
\end{figure}

\begin{figure}
    \centering
    \includegraphics[width=0.5\textwidth]{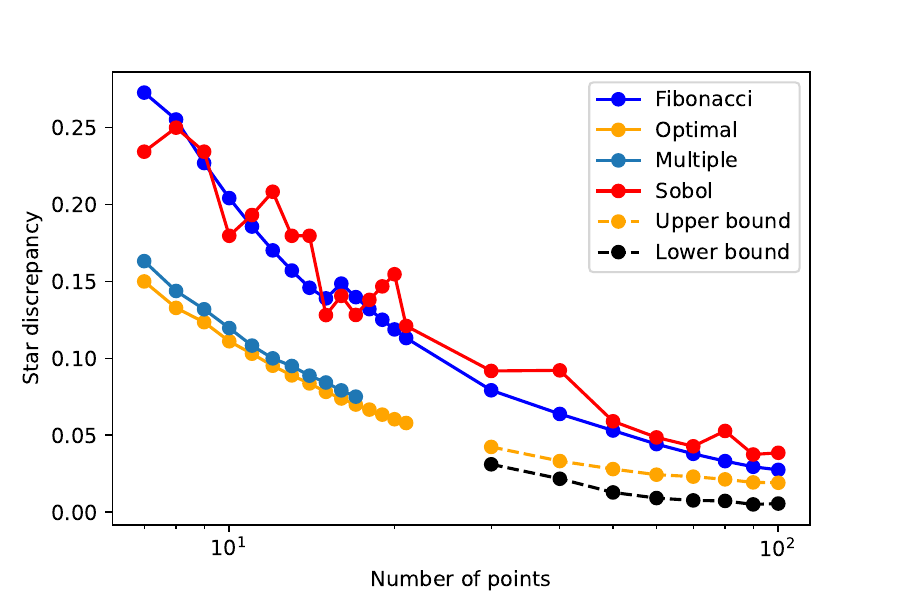}
    \caption{$L_{\infty}$ star discrepancies for the $L_{\infty}$ star optimal sets (line ``optimal'') and our multiple-corner optimal sets (line ``multiple''), compared to the Fibonacci set. The dashed lines are lower and upper bounds described in~\ref{tab:cont_heuristic}}.
    \label{fig:plotopt}
\end{figure}

Our framework also proves to be very flexible to tackle a number of related problems in the field of discrepancy, including the construction of point sets with optimal extreme discrepancy or periodic discrepancy. Details will be presented in Section \ref{sec:ext} and should not be seen as an exhaustive list of all the possibilities.

\textbf{Structure of the paper:} We recall in Section~\ref{sec:Gene} some background and important results on the $L_{\infty}$ star discrepancy. Section~\ref{sec:Formu} introduces our formulations as well as some lemmas on the structure of an optimal set. Some possible extensions of these formulations are presented in Section~\ref{sec:ext}, in particular the extension to higher dimensions. The optimal sets and their associated discrepancy values for the different settings are presented in the relevant sections.

\section{The \texorpdfstring{$L_{\infty}$}{L-infinity} star discrepancy}\label{sec:Gene}

\subsection{General results on the \texorpdfstring{$L_{\infty}$}{L-infinity} star discrepancy}

The $L_{\infty}$ star discrepancy of a point set represents the worst absolute difference between the volume of a box of the form $[0,q)$ and the proportion of points that fall inside this box. Formally, for a point set $P$ of $n$ points in $[0,1)^d$, the $L_{\infty}$ star discrepancy of $P$, $d^*_{\infty}(P)$ is defined by
$$
d^*_{\infty}(P) := \sup_{q\in[0,1)^d} \left| \frac{D(q,P)}{|P|} - \lambda(q) \right|,
$$
where $D(q,P)$ denotes the number of points of $P$ in the box $[0,q)$ and $\lambda(q)$ the Lebesgue measure of the $d$-dimensional box $[0,q)$. While the formal definition is with $P \subset [0,1)^d$, it is common to change to $[0,1]^d$ if more convenient and depending on the context \cite[Section 15.1]{OwenBook}. We will be using $[0,1]^d$ as it makes it easier to define the variables for the solver. The lower bound we prove in Section~\ref{sec:mingap} also shows that our sets will be in $[1/n,1]$, hence avoiding the main issue that may arise with periodicity and $0\equiv 1 \mod 1$.

As shown by Niederreiter in~\cite{NieBox}, computing the $L_{\infty}$ star discrepancy is a discrete problem, only points on a specific position grid can reach the maximal value. Figure \ref{fig:op_im} clearly shows these grids simply by doing local evaluations of the discrepancy values. First, while the definition only includes open boxes $[0,q)$, closed boxes [0,q] can be considered too.
Indeed, any closed anchored box in $[0,1]^d$ can be obtained as the limit of a sequence of bigger open boxes that contain the same number of points. The only exception is $[1,\ldots,1]$ and this closed box cannot give the worst discrepancy value as its local discrepancy is 0. More formally, for $P=\{x^{(1)},\dots,x^{(n)}\}$ the desired grid is given by

\begin{equation}
	\Gamma(P) := \Gamma_1(P) \times \ldots \times \Gamma_d(P)
\qquad \text{ and } 
\qquad 
	\overline{\Gamma}(P) := \overline{\Gamma}_1(P) \times \ldots \times \overline{\Gamma}_d(P),
\end{equation}
with
\begin{equation}
	\Gamma_j(P) := \{x_j^{(i)} | i \in 1,\ldots,n\} 
\qquad \text{ and } 
\qquad 
	\overline{\Gamma}_j(P) := \Gamma_j(P) \cup \{1\}.
\end{equation}

The $L_{\infty}$ star discrepancy computation thus reduces to the following discrete problem~\cite{DGWBook}

\begin{equation}
d_{\infty}^{*}(P) = \max \left\{\max_{q \in \overline{\Gamma}(P)}\frac{|P \cap [0,q]|}{|P|}-\lambda([0,q]),
\max_{q \in \Gamma(P)} \lambda([0,q))-\frac{|P \cap [0,q)|}{|P|}\right\}.
\label{discrepancy_formula}
\end{equation}

This formulation can be further simplified, as not all the points mentioned above can reach the maximum. For a box $[0,q)$ (or $[0,q]$) to be \emph{critical}, it is necessary for there to be, for every $j \in \{1,\ldots,d\}$, a point $x^{(i)} \in P$ such that $q_j=x_j^{(i)}$ and $x^{(i)} \leq q$ coordinate-wise. In other terms, each of the outer edges of the critical box must have at least one point on it. For any box $[0,q)$ or $[0,q]$ that doesn't fulfill this condition, it is possible to increase (for an open box) or decrease (for a closed box) one of the coordinates of $q$ without changing the number of points in the box and obtain a worse local discrepancy value. This distinction will be important for our models, as only these critical boxes will need to be considered.

\subsection{A generalization of a result  in~\texorpdfstring{\cite{whit:onop:1976}}{White} }\label{sec:White}

As mentioned in the introduction, White~\cite{whit:onop:1976} gave optimal discrepancy values for $n \leq 6$ in dimension 2. This work also includes a lower bound on the optimal $L_{\infty}$ star discrepancy values in dimension 2 for $n\geq 6$ (see Proposition~1 in \cite{whit:onop:1976}).

Given the relevance of Proposition 1 from~\cite{whit:onop:1976} for our models, we provide below a second proof of Proposition 1, that we furthermore generalize to any dimension. This result will be used to provide a lower bound constraint to the solver.

\begin{theorem}\cite[Proposition 1]{whit:onop:1976} \label{th:white}
Let $P \subset [0,1)^d$ with $|P|=n$. If $d=2$ and $n \geq 4$, or $d \geq 3$ and $n \geq 3$, then $d^*_{\infty}(P) \geq 1/n$.
\end{theorem}

\begin{proof}
   Let $n,d,P$ be as required.
There are two cases to consider: whether there exist $x$ and $y$ in $P$ such that neither $x \leq y$ nor $y \leq x$ coordinate-wise, or not. We will show that in both cases there exists a box with discrepancy at least $1/n$.
    \begin{itemize}
        \item We first suppose that the points in $P$ can be ordered such that $x^{(1)}\leq x^{(2)}\leq \ldots\leq x^{(n)}$. For all the large open boxes reaching an outer edge of $[0,1]^d$ to have discrepancy smaller than $1/n$, each coordinate $x_j^{(i)}$ has to be smaller than $i/n$. Given this, in the best possible case, the smallest closed box containing $x^{(i)}$ has volume smaller than $(i/n)^d$ and contains $i$ points. In dimension 2, for $i=2$ and $n \geq 4$, we then have a local discrepancy of the smallest closed box containing $x^{(2)}$ of more than $\frac{2}{n}-\frac{2^d}{n^d}$,  which is at least $1/n$. The same result holds in dimension at least 3 with the same value of $i$ and $n \geq 3$.
        \item If the points are not pairwise dominating 
        each other, there exist $x$ and $y$ in $P$ such that $x_1 > y_1$ and $x_2 < y_2$ (without loss of generality on the dimensions). The box $[0,q)$, where $q_1=x_1$ and $q_2=y_2$ and $q_j= \max(x_j,y_j)$ for $j \in \{3,\ldots,d\}$, contains at least two fewer points than $[0,q]$ and has the same Lebesgue measure $V_q$. Let $k$ be the number of points inside $[0,q)$. Then we have $d_{\infty}^{*}(P) \ge \max \{|V_q - k/n|, |V_q-(k+2)/n|\} \ge 1/n$. One of the two boxes therefore has discrepancy at least $1/n$.
        
    \end{itemize}
    We have thus shown that in both cases, there exists a box with local discrepancy at least $1/n$, which concludes the proof.
\end{proof}
By a similar argument as in Case~2, one can show that any set in dimension $d$ with $d$ mutually non-dominated points has discrepancy at least $d/(2n)$. We conjecture this to be a lower bound for all sets, provided $n$ is large enough while still being far smaller than the exponential number of points based on \cite{BilykSmall}.

\section{Problem formulations in two dimensions}\label{sec:Formu}

We introduce in this section our models to construct optimal $L_{\infty}$ star discrepancy sets. Section \ref{sec:2dcl} presents the more intuitive model, where we are directly optimizing point positions. It is in general the better performing model in 2 dimensions. In this model, we initially make an assumption of a minimal gap on the coordinates of two different points. Section \ref{sec:mingap} proves that this hypothesis does not increase the optimal $L_{\infty}$ star discrepancy value by showing not only that there exist optimal 
point sets with distinct coordinates but also that it is possible to set the minimal coordinate to $1/n$ when $n \geq 4$. It also describes a general procedure to shift points without increasing the discrepancy value of the set. Section \ref{sec:2dgrid} presents a second model, based on the fact that any point set naturally defines a grid $\Gamma$ as introduced in the previous section. Optimizing the placement of this grid and adding constraints enforcing exactly one point per grid line/column is equivalent to optimizing the placement of the point set. Finally, we present our results for these two models in Section~\ref{sec:results} and provide a visual comparison of our optimal sets to known low-discrepancy sets in Section~\ref{sec:images}.

\subsection{A ``classical'' formulation}\label{sec:2dcl}
We consider the problem of locating a set $P$ of $n$ points $x^{(i)}:=(x_{2i-1},x_{2i})\in[0,1]^2$ for $i=1,\dots,n$, such that $$P= \argmin_{X=\{x^{(1)},\dots,x^{(n)}\}}
d_{\infty}^*(X).$$ 
As mentioned above, we make the assumption that no two points in $P$ have the same value in any coordinate. 
Without loss of generality we assume that $x_{2i-1}<x_{2i+1}$ for all $i=1,\dots,n-1$, i.e., the points are labeled in increasing order of their first components. The point set $P$ then induces a grid $\overline{\Gamma}(P)$ with grid points $g_{ij}=(x_{2i-1},x_{2j})$ for $i,j=1,\dots,n$. By \eqref{discrepancy_formula}, to compute $d^*_{\infty}(P)$ for a point set $P$, we have to consider all critical grid points $g_{ij}=(x_{2i-1},x_{2j})$ for which the defining points $x^{(i)}$ and $x^{(j)}$ are located on or just below $g_{i,j}$. The first case corresponds to closed boxes, constraints~\eqref{eq:upperboundcont_dummy}, where $i=j$ is possible. The second case corresponds to open boxes, constraints~\eqref{eq:lowerboundcont_dummy}, where the points can also be defined by the outer edges of the box. Furthermore, to avoid having to treat the boxes with a coordinate equal to 1 as separate cases, we can add two dummy points $x^{(0)}=(0,1)$ and $x^{(n+1)}=(1,0)$, from which we only need the coordinates $x_0$ and $x_{2n+1}$. These have fixed values equal to $1$. We do not include them in any of the sums counting the number of points inside a box as they do not represent a real point of the final set.

For each of these critical boxes, we need to determine which points of $P$ are inside. While it is easy to decide whether $x^{(i)}$ is below $x^{(j)}$ in the first component by simply comparing the respective indices (recall that by assumption,  $x^{(i)}_1=x_{2i-1}<x_{2j-1}=x^{(j)}_1 $ if and only if $i<j$), we need to define 
indicator variables $y_{ij}\in\{0,1\}$ for all $i,j\in\{1,\dots,n\}$ to indicate that
$x^{(i)}$ is below, or equal to, $x^{(j)}$ in the second component (note that equality is only possible if $i=j$). In other words, we want that $y_{ij}=1$ if and only if $x^{(i)}_2=x_{2i} \leq x_{2j}=x^{(j)}_2$. Then, 
a grid point $g_{ij}$ needs to be considered whenever $j \leq 
i$ and $y_{ij}=1$ (closed boxes), and whenever $j<i$ and $y_{ij}=1$ (open boxes).

In order to properly define the indicator variables $y_{ij}$ for all $i,j\in\{1,\dots,n\}$, we first consider the case that $i<j$. 
Towards this end, observe that $x_{2i}> x_{2j}$ and hence $y_{ij}=0$ if and only if $x_{2j}-x_{2i}<0$. This can be translated into linear constraints on the indicator variable $y_{ij}$ as follows: 
\begin{equation*}
    \begin{array}{rclcrcl}
    \left(x_{2j}-x_{2i} < 0 \right. & \Rightarrow & \left. y_{ij}=0\right) &\quad \text{is enforced by}\quad & x_{2j}-x_{2i} &>& y_{ij}-1\\
    \left(x_{2j}-x_{2i} < 0 \right. & \Leftarrow & \left. y_{ij}=0\right) &\quad \text{is enforced by}\quad & x_{2j}-x_{2i} &<& y_{ij}.
    \end{array}
\end{equation*}
By anti-symmetry, we have $y_{ij}=1 - y_{ji}$ for all $i,j=1,\dots,n$, $i\neq j$. Moreover, we set $y_{ii}=1$ for all $i=1,\dots,n$. These constraints correspond to constraints \eqref{eq:yone} to \eqref{eq:yfour} in the model below. Note that in the formulation below, we additionally require a small distance $\varepsilon>0$ between $x$-coordinate (constraints \eqref{eq:mindist}) and $y$-coordinate values (constraints \eqref{eq:yone}) of distinct points in order to avoid degenerate situations with different points with one or multiple equal coordinates. We show in Section~\ref{sec:mingap} that there is at least one optimal solution verifying this for $\varepsilon$ small enough. A solution to the model with this small enough $\varepsilon$ is then a provably optimal solution of our initial problem.

We hence obtain the following nonlinear programming problem (that we refer to as model \eqref{eq:2dcont_dummy} in the following) that has quadratic terms in the constraints due to the volume computations:

\setcounter{equation}{4}
\begin{subequations}\label{eq:2dcont_dummy}
\begin{align}
    \min\;\; & f && \nonumber\\ 
    \text{s.t.}\;\; & \displaystyle \frac{1}{n} \sum_{u=1}^i y_{uj} - x_{2i-1}x_{2j} \leq f \!+\! (1\!-\!y_{ij}) && \forall i,j=1,\dots,n,\, j\leq i \label{eq:upperboundcont_dummy}\\
    & \displaystyle \frac{-1}{n}\left(\sum_{u=0}^{i-1} y_{uj}-\!1\right) + x_{2i-1}x_{2j} \leq f \!+\! (1\!-\!y_{ij})&& \forall i=1,\dots,n\!+\!1,\, j<i
    \label{eq:lowerboundcont_dummy}\\
    \setcounter{equation}{2}
    & x_{2i+1}-x_{2i-1} \geq \varepsilon && \forall i=1,\dots,n-1 \label{eq:mindist}\\
    & x_{2j}-x_{2i} \geq y_{ij}-1 + \varepsilon && \forall i,j=1,\dots,n,\, i<j \label{eq:yone}\\
    & x_{2j}-x_{2i} \leq y_{ij} && \forall i,j=1,\dots,n,\, i<j \label{eq:ytwo}\\
    & y_{ij}=1-y_{ji} && \forall i,j=1,\dots,n,\, i>j \label{eq:ythree}\\
    & y_{ii}=1 && \forall i=1,\dots,n \label{eq:yfour}\\
    & x_0=x_{2n+1}=1 && \nonumber \\
    &y_{0j}=0 && \forall j=1,\dots,n\nonumber \\
    &y_{j0}=y_{(n+1),j}=1 && \forall j=0,\dots,n \nonumber \\ 
    & x_{i}\in(0,1] && \forall i=1,\dots,2n\nonumber\\ 
    &y_{ij}\in\{0,1\}&& \forall i,j=1,\dots,n\nonumber \\ &f\geq 0 .\hspace*{-4cm} &&\nonumber
\end{align}
\end{subequations}
Note that the variables $y_{j,(n+1)}$ for $j=0,\dots,n+1$ are never used in the model and are hence not defined.

As mentioned previously, constraints \eqref{eq:upperboundcont_dummy} and \eqref{eq:lowerboundcont_dummy} only need to be enforced for critical grid points. We thus enforce them only for $j\leq i$ for closed boxes, accounting for the case that a point defines a critical grid point by itself, and for $j<i$ in the case of open boxes. Moreover, the constraints are relaxed (by adding $1$ on the right-hand side) whenever $x^{(j)}$ is not above (or equal to) $x^{(i)}$ (and hence $x^{(i)}$ is not below $g_{ij}$ and the box is not critical), i.e., whenever $y_{ij}=0$. For constraints \eqref{eq:lowerboundcont_dummy}, we need to count all points in the respective volume that are strictly below the considered grid point $g_{ij}$. Since we assume that no two points share the same coordinate value in any dimension, the summation is over the indices $1$ to $i-1$, where we have to correct for the case that $u=j$ for which we consider the point $x^{(j)}$ with $y_{jj}=1$ by subtracting $1$ from the sum over the $y_{uj}$'s. Boxes with a coordinate equal to $1$ correspond to the special cases where $i=n+1$ or $j=0$ in constraints \eqref{eq:lowerboundcont_dummy}.

Constraints \eqref{eq:mindist} (for the first dimension) and \eqref{eq:yone} and \eqref{eq:ytwo} (for the second dimension) impose a minimum coordinate difference $\varepsilon$. Since the points are ordered in the first dimension, only consecutive pairs need to be checked there. Section~\ref{sec:mingap} justifies this minimal coordinate difference.

In order to strengthen model \eqref{eq:2dcont_dummy}, additional constraints may be added.
\setcounter{equation}{4}
\begin{subequations}\label{eq:2dcont_extra}
\setcounter{equation}{7}
\begin{align}
    & f \geq 1/n && \text{valid if $n \geq 4$} 
    \label{eq:2dcont_boundz}\\ 
    & y_{ij}+y_{jk} -1 \leq y_{ik} && \forall i,j,k=1,\dots,n \label{eq:2dcont_triangle}\\
    & y_{ij}+y_{jk} \geq y_{ik} && \forall i,j,k=1,\dots,n \label{eq:2dcont_triangle2}\\
    & \sum_{i=1}^n \sum_{j=1}^n y_{ij} = \frac{n(n+1)}{2} && \label{eq:2dcont_sum}
\end{align}
\end{subequations}

This includes, among others, known bounds on the optimal value of $f$ as well as constraints that are based on the specific properties of the sorting variables $y_{ij}$ as, e.g., transitivity.
\cite{whit:onop:1976} derived a lower bound $d^*_{\infty} \geq 1/n$ value for optimal points sets with $n> 6$ in dimension 2, and that we generalize to $n>4$ in any dimension in Theorem~\ref{th:white}. This knowledge can be included in the model by adding constraint \eqref{eq:2dcont_boundz}. Constraints \eqref{eq:2dcont_triangle} enforce transitivity for the indicator variables, i.e., if $x_2^{(i)}=x_{2i}<x_{2j}=x_2^{(j)}$ and $x^{(j)}_2=x_{2j}< x_{2k}=x^{(k)}_2$
(i.e., if $y_{ij}=1$ and $y_{jk}=1$), then also $x_2^{(i)}=x_{2i}<x_{2k}=x_2^{(k)}$ (i.e., $y_{ik}=1$). Note that constraints \eqref{eq:2dcont_triangle} are also valid for cases where some or all of the indices $i,j,k$ are equal. Constraints \eqref{eq:2dcont_triangle2} cover the converse case where 
$x_2^{(i)}>x_2^{(j)}$ and $x^{(j)}_2>x^{(k)}_2$
(i.e., if $y_{ij}=0$ and $y_{jk}=0$) and hence  $x_2^{(i)}>x_2^{(k)}$ (i.e., $y_{ik}=0$). 
Finally, constraint \eqref{eq:2dcont_sum} counts the number of $y_{ij}$ variables that have the value $1$ in any feasible solution, which is always constant.

\subsection{Minimal point spacing}\label{sec:mingap}

In order to refine the model further and justify our general position hypothesis, we prove in this section some properties on  optimal point sets. In particular we show that there exists some optimal sets in dimension 2 such that
\begin{itemize}
    \item No two points share the same coordinate value in any of the coordinates (we refer to this as the \emph{general position assumption}).
    \item Either the points all have a minimum distance from the lower and left boundaries (up and right shifts, Lemma~\ref{lemma:upshift_lb}), or the respectively first points are on the lower and left boundaries of the unit cube (down and left shifts, Lemma~\ref{lemma:downshift_ub}). This results holds in any dimension.
    \item The points satisfy some minimum distance requirements, i.e., they can be moved such that the distance between two vertically consecutive points / horizontally consecutive points is at least $\varepsilon$ with a sufficiently small $\varepsilon>0$. 
\end{itemize}

\begin{definition}[up-shift, right-shift]
    Consider an $n$-point set $P=\{x^{(1)},\dots,x^{(n)}\}$ in $[0,1]^2$ and let $i\in\{1,\dots,n\}$.
    \begin{enumerate}
        \item The movement of a point $x^{(i)}\in P$ is called an \emph{up-shift}, if the point $x^{(i)}$ is moved vertically upwards. An up-shift leads to a new $n$-point set $\up(P,x^{(i)},\delta)=P\setminus \{x^{(i)}\} \cup \{(x^{(i)}_1,x^{(i)}_2+\delta)\}$, with $0< \delta\leq 1-x^{(i)}_2$. An up-shift is called \emph{admissible}, if 
        $$\delta\leq\max\left\{0, \, \frac{1}{n}-\min_{k\neq i: x^{(k)}\leq x^{(i)}}
        x^{(i)}_2-x^{(k)}_2\right\}.$$ 
        \item Analogously, the movement of a point $x^{(i)}\in P$ is called a \emph{right-shift}, if the point $x^{(i)}$ is moved horizontally to the right. A right-shift leads to a new $n$-point set $\ri(P,x^{(i)},\delta)=P\setminus \{x^{(i)}\} \cup \{(x^{(i)}_1+\delta,x^{(i)}_2)\}$, with $0<\delta\leq 1-x^{(i)}_1$. A right-shift is called \emph{admissible}, if 
        $$\delta\leq\max\left\{0,\, \frac{1}{n}-\min_{k\neq i: x^{(k)}\leq x^{(i)}}
        x^{(i)}_1-x^{(k)}_1\right\}.$$ 
    \end{enumerate}
\end{definition}

\begin{lemma}\label{lemma:upshift_ub}
Let $P^*$ be an optimal $n$-point set with $d^*_{\infty}(P^*)=f^*$. Then, an up-shift never violates a closed box constraint \eqref{eq:upperboundcont_dummy} (that defines an upper bound on the number of points in a closed box) w.r.t.\ $f^*$. Analogously, a right-shift never violates a closed box constraint \eqref{eq:upperboundcont_dummy} w.r.t.\ $f^*$.
\end{lemma}
\begin{proof}
    Consider an up-shift from $P^*$ to a new $n$-point set $\hat{P}:=\up(P^*,x^{(i)},\delta)$ with $i\in\{1,\dots,n\}$ arbitrary but fixed. We refer to the grid points induced by $P^*$ as $\Gamma(P^*):=\{g_{ab}:\,a,b=1,\dots,n+1\}$ and to the grid points induced by $\hat{P}$ by $\Gamma(\hat{P}):=\{\hat{g}_{ab}:\,a,b=1,\dots,n+1\}$, where the index $n+1$ refers to the respective grid points on the upper and right boundary of $[0,1]^2$, respectively.

    First note that $\hat{g}_{ab}=g_{ab}$ for all $a,b=1,\dots,n+1$, $b\neq i$, and either the number of points in the closed box $[0,\hat{g}_{ab}]$ stayed the same or decreased by one. Hence, the left-hand term of the closed box constraint \eqref{eq:upperboundcont_dummy} is the same or decreased, the constraint remains satisfied. See Figure~\ref{fig:shifting} for an illustration, where $i=8$ and the grid points 
    $\hat{g}_{a8}$ are highlighted by empty orange circles. Moreover, all grid points $\hat{g}_{ai}$ cover a larger area after the shifting than the corresponding grid points $g_{ai}$, for $i=1,\dots,n+1$, while the number of covered points remains unchanged as long as no other horizontal grid line is crossed. If a horizontal grid line is crossed, then also a grid point $g_{aj}$ with $j\neq i$ is crossed at which the closed box constraint \eqref{eq:upperboundcont_dummy} was satisfied even when counting the point $x^{(i)}$ in $P^*$, and again we have a larger area covered by $\hat{g}_{ai}$ in $\hat{P}$. Hence, the closed box constraints \eqref{eq:upperboundcont_dummy} remain satisfied at these points as in the regular $\hat{g}_{ab}$ case. 
    
    The argumentation for right-shifts is completely analogous since both dimensions have a symmetric role.
\end{proof}

Note that the proof of Lemma~\ref{lemma:upshift_ub} does not rely on the admissibility of the up-shift or right-shift, respectively. Given that upper bound constraints \eqref{eq:upperboundcont_dummy} represent closed boxes for which we prefer fewer points, reducing the number of points in a fixed volume naturally doesn't increase the $L_{\infty}$ star discrepancy for the closed box.

\begin{lemma}\label{lemma:upshift_lb}
Let $P^*$ be an optimal $n$-point set with $d^*_{\infty}(P^*)=f^*$. Then, an \textbf{admissible} up-shift never violates an open box constraint \eqref{eq:lowerboundcont_dummy} (that defines a lower bound on the number of points in an open box) w.r.t.\ $f^*$. Analogously, an admissible right-shift never violates an open box constraint \eqref{eq:lowerboundcont_dummy} w.r.t.\ $f^*$.
\end{lemma}

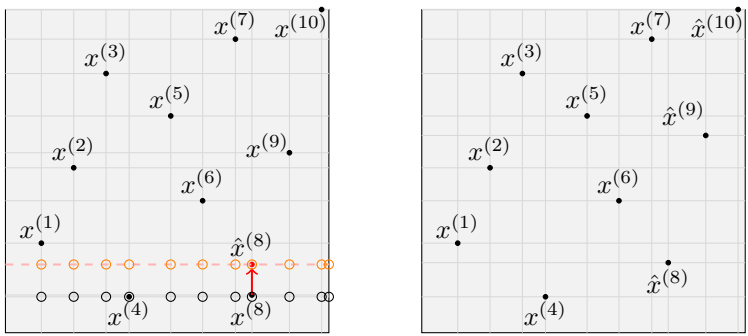
\begin{figure}[t]
\centering
\begin{tikzpicture}[scale=0.43]
\fill[gray!10] (0,0)rectangle(10,10);
\draw (0,0)rectangle(10,10);
\draw[gray!30] (0,2.7702) -- (10,2.7702);
\draw[gray!30] (0,5.1030) -- (10,5.1030);
\draw[gray!30] (0,8.0191) -- (10,8.0191);
\draw[gray!30] (0,1.1111) -- (10,1.1111);
\draw[gray!30] (0,6.7068) -- (10,6.7068);
\draw[gray!30] (0,4.0824) -- (10,4.0824);
\draw[gray!30] (0,9.0829) -- (10,9.0829);
\draw[gray!30] (0,1.1664) -- (10,1.1664);
\draw[gray!30] (0,5.5680) -- (10,5.5680);
\draw[gray!30] (0,10.0) -- (10,10.0);
\draw[gray!30] (1.1111,0) -- (1.1111,10);
\draw[gray!30] (2.1111,0) -- (2.1111,10);
\draw[gray!30] (3.1111,0) -- (3.1111,10);
\draw[gray!30] (3.8248,0) -- (3.8248,10);
\draw[gray!30] (5.1111,0) -- (5.1111,10);
\draw[gray!30] (6.0966,0) -- (6.0966,10);
\draw[gray!30] (7.1111,0) -- (7.1111,10);
\draw[gray!30] (7.6207,0) -- (7.6207,10);
\draw[gray!30] (8.7804,0) -- (8.7804,10);
\draw[gray!30] (9.7757,0) -- (9.7757,10);
\filldraw (1.1111,2.7702) circle (2pt);
\node at (1.1111,3.3) {$x^{(1)}$};
\filldraw (2.1111,5.1030) circle (2pt);
\node at (2.1111,5.7) {$x^{(2)}$};
\filldraw (3.1111,8.0191) circle (2pt);
\node at (3.1111,8.6) {$x^{(3)}$};
\filldraw (3.8248,1.1111) circle (2pt);
\node at (3.8248,0.6) {$x^{(4)}$};
\filldraw (5.1111,6.7068) circle (2pt);
\node at (5.1111,7.3) {$x^{(5)}$};
\filldraw (6.0966,4.0824) circle (2pt);
\node at (6.0966,4.7) {$x^{(6)}$};
\filldraw (7.1111,9.0829) circle (2pt);
\node at (7.1111,9.6) {$x^{(7)}$};
\filldraw (7.6207,1.1664) circle (2pt);
\node at (7.6207,0.6) {$x^{(8)}$};
\draw[color=red,line width=0.75pt,->] (7.6207,1.1664) -- (7.6207,2.0); 
\filldraw[color=red] (7.6207,2.1111) circle (2pt);
\node at (7.6207,2.75) {$\hat{x}^{(8)}$};
\draw[red!30, dashed, line width=0.75pt] (0,2.1111) -- (10,2.1111);
\filldraw (8.7804,5.5680) circle (2pt);
\node at (8.1,5.8) {$x^{(9)}$};
\filldraw (9.7757,10.0) circle (2pt);
\node at (9.15,9.5) {$x^{(10)}$};
\draw (1.1111,1.1111) circle (4pt);
\draw (2.1111,1.1111) circle (4pt);
\draw (3.1111,1.1111) circle (4pt);
\draw (3.8248,1.1111) circle (4pt);
\draw (5.1111,1.1111) circle (4pt);
\draw (6.0966,1.1111) circle (4pt);
\draw (7.1111,1.1111) circle (4pt);
\draw (7.6207,1.1111) circle (4pt);
\draw (8.7804,1.1111) circle (4pt);
\draw (9.7757,1.1111) circle (4pt);
\draw (10,1.1111) circle (4pt);
\draw[color=orange] (1.1111,2.1111) circle (4pt);
\draw[color=orange] (2.1111,2.1111) circle (4pt);
\draw[color=orange] (3.1111,2.1111) circle (4pt);
\draw[color=orange] (3.8248,2.1111) circle (4pt);
\draw[color=orange] (5.1111,2.1111) circle (4pt);
\draw[color=orange] (6.0966,2.1111) circle (4pt);
\draw[color=orange] (7.1111,2.1111) circle (4pt);
\draw[color=orange] (7.6207,2.1111) circle (4pt);
\draw[color=orange] (8.7804,2.1111) circle (4pt);
\draw[color=orange] (9.7757,2.1111) circle (4pt);
\draw[color=orange] (10,2.1111) circle (4pt);
\end{tikzpicture}\hspace{1cm}
\begin{tikzpicture}[scale=0.43]
\fill[gray!10] (0,0)rectangle(10,10);
\draw (0,0)rectangle(10,10);
\draw[gray!30] (0,2.7702) -- (10,2.7702);
\draw[gray!30] (0,5.1030) -- (10,5.1030);
\draw[gray!30] (0,8.0191) -- (10,8.0191);
\draw[gray!30] (0,1.1111) -- (10,1.1111);
\draw[gray!30] (0,6.7068) -- (10,6.7068);
\draw[gray!30] (0,4.0824) -- (10,4.0824);
\draw[gray!30] (0,9.0829) -- (10,9.0829);
\draw[gray!30] (0,2.1664) -- (10,2.1664);
\draw[gray!30] (0,6.1030) -- (10,6.1030);
\draw[gray!30] (0,10.0) -- (10,10.0);
\draw[gray!30] (1.1111,0) -- (1.1111,10);
\draw[gray!30] (2.1111,0) -- (2.1111,10);
\draw[gray!30] (3.1111,0) -- (3.1111,10);
\draw[gray!30] (3.8248,0) -- (3.8248,10);
\draw[gray!30] (5.1111,0) -- (5.1111,10);
\draw[gray!30] (6.0966,0) -- (6.0966,10);
\draw[gray!30] (7.1111,0) -- (7.1111,10);
\draw[gray!30] (7.6207,0) -- (7.6207,10);
\draw[gray!30] (8.7804,0) -- (8.7804,10);
\draw[gray!30] (9.7804,0) -- (9.7804,10);
\filldraw (1.1111,2.7702) circle (2pt);
\node at (1.1111,3.3) {$x^{(1)}$};
\filldraw (2.1111,5.1030) circle (2pt);
\node at (2.1111,5.7) {$x^{(2)}$};
\filldraw (3.1111,8.0191) circle (2pt);
\node at (3.1111,8.6) {$x^{(3)}$};
\filldraw (3.8248,1.1111) circle (2pt);
\node at (3.8248,0.6) {$x^{(4)}$};
\filldraw (5.1111,6.7068) circle (2pt);
\node at (5.1111,7.3) {$x^{(5)}$};
\filldraw (6.0966,4.0824) circle (2pt);
\node at (6.0966,4.7) {$x^{(6)}$};
\filldraw (7.1111,9.0829) circle (2pt);
\node at (7.1111,9.6) {$x^{(7)}$};
\filldraw (7.6207,2.1664) circle (2pt);
\node at (7.6207,1.6) {$\hat{x}^{(8)}$};
\filldraw (8.7804,6.1030) circle (2pt);
\node at (8.1,6.8) {$\hat{x}^{(9)}$};
\filldraw (9.7804,10.0) circle (2pt);
\node at (9.15,9.5) {$\hat{x}^{(10)}$};
\end{tikzpicture}

\caption{Left: An optimal $10$-point set $P^*$ with $f^*=d^*_{\infty}(P^*)=0.1111$. The red arrow indicates the up-shift $\up(P^*,x^{(8)},0.09447)$, which is admissible with $\delta=\frac{1}{n}-(x^{(8)}_2-x^{(4)}_2)$. Note that $x^{(8)}$ is slightly higher than $x^{(4)}$, i.e., the two points are not on the same horizontal grid line. Right: Alternative optimal $10$-point set after implementing all admissible up-shifts and all admissible right-shifts.\label{fig:shifting}}
\end{figure}

\begin{proof}First consider an admissible up-shift and note that if $\delta=0$ (and hence $x^{(i)}=\hat{x}^{(i)}$) or if $x^{(i)}_1=1$ (and hence $\hat{x}^{(i)}_1=1$), then the result is trivial.

Now let $\delta>0$ and $x^{(i)}_1=\hat{x}^{(i)}_1<1$.
We use the same notation as in the proof of Lemma~\ref{lemma:upshift_ub}  and consider grid points $\hat{g}_{ab}$ for all $a,b=1,\dots,n+1$. 
If neither the number of points in the open box $[0,\hat{g}_{ab})$ nor its volume is changed as compared to $[0,g_{ab})$, then as in the proof of Lemma~\ref{lemma:upshift_ub} the corresponding open box constraint \eqref{eq:lowerboundcont_dummy} remains satisfied. 
    
In particular, this applies to all grid points $\hat{g}_{ak}$, $a=1,\dots,n+1$, for which $x^{(k)}\leq x^{(i)}$, $k\neq i$ (since we consider an up-shift of $x^{(i)}$). Now observe that every open box $[0,\hat{g}_{ab})$ with $x^{(k)}<\hat{g}_{ab}$ contains at least one additional point as compared to $\hat{g}_{ak}$, namely the point $x^{(k)}$. This implies that, independent of the location of other points in $\hat{P}$, the open box constraint \eqref{eq:lowerboundcont_dummy} is satisfied at all such grid points $\hat{g}_{ab}$ with $x^{(k)}<\hat{g}_{ab}$ for which the additional volume, as compared to that covered by $\hat{g}_{ak}$, does not exceed the value of $\frac{1}{n}$. This clearly includes all grid points $\hat{g}_{ab}$ in the box $B_{1,k}:=(x^{(k)},(1,x^{(k)}_2+\frac{1}{n})]$ with lower left corner at $x^{(k)}$ and upper right corner at $(1,x^{(k)}_2+\frac{1}{n})$. Let $B_1:=\bigcup_{k\neq i: x^{(k)}\leq x^{(i)}} B_{1,k}$.
    
Moreover, the only grid points $\hat{g}_{ab}$ for which the number of points in the open box $[0,\hat{g}_{ab})$ may change are the grid points in the box $B_2:=(x^{(i)},(1,\hat{x}^{(i)}_2)]$ with lower left corner at $x^{(i)}$ and upper right corner at $(1,\hat{x}^{(i)}_2)$. Similarly, the only grid points $\hat{g}_{ab}$ for which the volume of the open box $[0,\hat{g}_{ab})$ may change are the grid points $\hat{g}_{ai}$, where the largest volume for the same number of points is attained for $a=n+1$, i.e., at the grid point $\hat{g}_{(n+1),i}=(1,\hat{x}^{(i)}_2)$ that is also  contained in the box $B_2$. Since we only consider admissible up-shifts, we have $B_2\subseteq B_1$ and the result follows.
 
\end{proof}

To rephrase Lemma~\ref{lemma:upshift_lb}, it is possible to shift a point $x\in P$ up or to the right if there is another point below and left, close enough to ensure that the open box constraint defined by $x$ is not active (i.e. the inequality is strict in the constraint), ``how close'' then defines by how much it is possible to shift $x$.

Similarly, we can define down-shifts and left-shifts in a very similar way, providing even more possibilities for modifying an existing point set.

\begin{definition}[down-shift, left-shift]
    Consider an $n$-point set $P=\{x^{(1)},\dots,x^{(n)}\}$ and let $i\in\{1,\dots,n\}$.
    \begin{enumerate}
        \item The movement of a point $x^{(i)}\in P$ is called a \emph{down-shift}, if the point $x^{(i)}$ is moved vertically downwards. A down-shift leads to a new $n$-point set $\down(P,x^{(i)},\delta)=P\setminus \{x^{(i)}\} \cup \{(x^{(i)}_1,x^{(i)}_2-\delta)\}$, with $0< \delta\leq x^{(i)}_2$. A down-shift is called \emph{admissible}, if 
        $$\delta\leq\max\left\{0,\, \frac{1}{n}-\min_{k\neq i: x^{(k)}\geq x^{(i)}}
        x^{(k)}_2-x^{(i)}_2\right\}.$$ 
        \item The movement of a point $x^{(i)}\in P$ is called a \emph{left-shift}, if the point $x^{(i)}$ is moved horizontally to the left. A left-shift leads to a new $n$-point set $\lef(P,x^{(i)},\delta)=P\setminus \{x^{(i)}\} \cup \{(x^{(i)}_1-\delta,x^{(i)}_2)\}$, with $0<\delta\leq x^{(i)}_1$. A left-shift is called \emph{admissible}, if 
        $$\delta\leq\max\left\{0,\, \frac{1}{n}-\min_{k\neq i: x^{(k)}\geq x^{(i)}}
        x^{(k)}_1-x^{(i)}_1\right\}.$$ 
    \end{enumerate}
\end{definition}

\begin{lemma}\label{lemma:downshift_lb}
Let $P^*$ be an optimal $n$-point set with $d^*_{\infty}(P^*)=f^*$. Then, a down-shift never violates an open box constraint \eqref{eq:lowerboundcont_dummy} w.r.t.\ $f^*$. Analogously, a left-shift never violates an open box constraint \eqref{eq:lowerboundcont_dummy} w.r.t.\ $f^*$.
\end{lemma}

\begin{proof}
    The result can be proven completely analogous to Lemma~\ref{lemma:upshift_ub}.
\end{proof}

\begin{lemma}\label{lemma:downshift_ub}
Let $P^*$ be an optimal $n$-point set with $d^*_{\infty}(P^*)=f^*$. Then, an admissible down-shift never violates a closed box constraint~\eqref{eq:upperboundcont_dummy} w.r.t.\ $f^*$. Analogously, an admissible left-shift never violates a closed box constraint \eqref{eq:upperboundcont_dummy} w.r.t.\ $f^*$.
\end{lemma}

\begin{proof}
    The result can be proven completely analogous to Lemma~\ref{lemma:upshift_lb}.
\end{proof}

As in Lemma~\ref{lemma:upshift_lb}, Lemma~\ref{lemma:downshift_ub} ensures that there is some other box defining a worse discrepancy value and that, with our choice of shift distance, no new box can have a discrepancy surpassing this worse value (which is not necessarily the worst).

Lemmas~\ref{lemma:upshift_ub} to~\ref{lemma:downshift_ub} imply that there always exists an optimal $n$-point set that satisfies additional ``distance'' constraints, while, of course, not all optimal $n$-point sets necessarily satisfy such constraints. In addition, existing $n$-point sets can be modified without deteriorating the objective value by applying up and right-shifts in any sequence. This implies, among others, that the points that are closest to the lower and left boundary of the unit square can be moved up and right, respectively. Since there are no further points below them, they can actually be moved up to a distance of $f$ from the respective boundaries, without increasing the discrepancy value $f$ of the considered configuration. Moreover, for every pair of points $x^{(i)}\neq x^{(j)}$ with $x^{(i)}_1\leq x^{(j)}_1$
and $x^{(i)}_2\leq x^{(j)}_2$, we can apply up or down-shifts so that the vertical distance between the two points is at least $\frac{1}{n}$, and similarly we can apply right- or left-shifts so that the horizontal distance between the two points is at least $\frac{1}{n}$. Note that this is also possible when one (or both) of the points are close to the upper or right boundary of the unit cube, since we can apply down or left-shifts in this case. 
The constraints~\eqref{eq:2dcont_boundx1}-\eqref{eq:4gstrong} implement these conditions, constraints \eqref{eq:4s} and \eqref{eq:4t} follow from the inequalities for the grid points on the upper boundary, and constraint \eqref{eq:4w} follows by the triangle inequality. 
\setcounter{equation}{4}
\begin{subequations}\label{eq:2dcontshift}
\setcounter{equation}{11}
\begin{align}
    & x_1 = f &&
    \label{eq:2dcont_boundx1}\\
    & x_{2i} \geq f && \forall i=1,\dots,n
    \label{eq:2dcont_boundx2i}\\
    & x_{2i+1}-x_{2i-1} \geq y_{i,(i+1)}-1+\frac{1}{n} && \forall i=1,\dots,n-1 \label{eq:4fstrong}\\
    & x_{2j}-x_{2i} \geq y_{ij}-1+\frac{1}{n} && \forall i,j=1,\dots,n,\, i<j\label{eq:4gstrong}\\
    & x_{2i-1} \leq f + \frac{i-1}{n} && \forall i=1,\ldots,n \label{eq:4s}\\
    & x_{2i-1} \geq \frac{i}{n}-f && \forall i=1,\dots,n \label{eq:4t}\\
    & x_{2j-1}-x_{2i-1} + x_{2j}-x_{2i} \geq y_{ij} - 1 + \frac{2}{n} 
    && \forall i,j=1,\dots,n-1,\, i<j \label{eq:4w}
\end{align}
\end{subequations}

Note that while constraints \eqref{eq:4fstrong} can be introduced in addition to constraints \eqref{eq:mindist} but can not replace them, constraints \eqref{eq:4gstrong} strengthen constraints \eqref{eq:yone} and \eqref{eq:ytwo} and should thus be used instead. Finally, we note that constraints \eqref{eq:4s} make some of the constraints of type \eqref{eq:lowerboundcont_dummy} redundant. Indeed, when $j=0$, constraint \eqref{eq:lowerboundcont_dummy} gives us exactly constraint \eqref{eq:4s}. With this improvement, we can now set $y_{i0}=0$ for all $i \in \{1,\ldots,n+1\}$.

\subsection{An assignment formulation}\label{sec:2dgrid}

We now introduce a second formulation, which will also exploit the results from Lemmas~\ref{lemma:upshift_ub} to~\ref{lemma:downshift_ub}. An alternative formulation is obtained when defining the two coordinates of all points in $P$ separately and independently from each other in two vectors $x,y\in[0,1]^n$. The point set $P$ is then obtained by assigning exactly one $y$-coordinate to every $x$-coordinate and vice versa, using assignment variables $a_{ij}\in\{0,1\}$, $i,j=1,\dots,n$. In other words, rather than optimizing a point set, we optimize the associated grid (and indirectly the set): a point $(x_i,y_j)$ is in $P$ if and only if $a_{ij}=1$. The advantage of this formulation is its simplicity, since the sorting can be implemented already when defining the $x$- and $y$-coordinates of the vertical and horizontal grid lines, respectively, in an increasing order.

\begin{subequations}\label{eq:M5_2dcont}
\begin{align}
    \min\;\; & f && \nonumber\\ 
    \text{s.t.}\;\; & \displaystyle \frac{1}{n} \sum_{u=1}^i\sum_{v=1}^j a_{uv} - x_{i}y_{j} \leq f && \forall i,j=1,\dots,n \label{eq:M5_upperboundcont}\\
    & \displaystyle \frac{-1}{n} \sum_{u=1}^{i-1}\sum_{v=1}^{j-1} a_{uv} + x_{i}y_{j} \leq f && \forall i,j=1,\dots,n+1 \label{eq:M5_lowerboundcont}\\
    & x_{n+1}=1,\; y_{n+1}=1 && \label{eq:M5_dummy}\\
    & x_{i+1}-x_{i} \geq \varepsilon && \forall i=1,\dots,n-1 \label{eq:M5_4f}\\
    & y_{i+1}-y_{i} \geq \varepsilon && \forall i=1,\dots,n-1 \label{eq:M5_4fbis}\\
    & \sum_{i=1}^n a_{ij} = 1 && \forall j=1,\dots,n \label{eq:M5_column}\\
    & \sum_{j=1}^n a_{ij} = 1 && \forall i=1,\dots,n \label{eq:M5_row}\\
    & x_i,y_i\in[0,1] && \forall i=1,\dots,n\; \nonumber \\
    & a_{ij}\in\{0,1\} && \forall i,j=1,\dots,n\; \nonumber \\
    & f\geq 0 .\hspace*{-4cm} &&\nonumber
\end{align}
\end{subequations}

Constraints \eqref{eq:M5_upperboundcont} and \eqref{eq:M5_lowerboundcont} correspond to the discrepancy inequalities, with the double sum counting the number of selected points inside the box defined by $(x_i,y_j)$. Constraints \eqref{eq:M5_4f} and \eqref{eq:M5_4fbis} impose the minimal distance between two grid lines, as derived in the previous subsection. Finally, constraints \eqref{eq:M5_column} and \eqref{eq:M5_row} impose that there is exactly one point in each column and row of the grid. The ``exactly'' comes from the general position assumption we derived previously.

Note that the values of $x_{n+1},y_{n+1}$ are fixed to one. These parameters are used as dummy values to include the open box constraints \eqref{eq:M5_lowerboundcont} for the grid points on the upper and right boundaries in a simple way.
Note also that many of the constraints \eqref{eq:M5_upperboundcont} and \eqref{eq:M5_lowerboundcont} are redundant since the defining points may not be located below or on the respective grid lines. It is possible to check if a box defined by $x_i$ and $y_j$ is critical by calculating $\sum_{k=1}^{i}a_{kj}+\sum_{k=1}^{j}a_{ik}$. This sum is equal to 2 if and only if there is a point on each of the outer edges and thus if and only if the box is critical. In practice, the model was slower when adding this requirement. It may be worthwhile to use this requirement for other solvers or different settings. 

As in the previous model, the formulation can be strengthened by the following constraints:
\setcounter{equation}{5}
\begin{subequations}
\setcounter{equation}{7}
\begin{align}
    & f\geq 1/n && \text{if $n\geq 4$} \label{eq:white}\\ 
    & x_1 = f \label{eq:M6_o}&&\\
    & y_1 = f &&  \\
    & x_j-x_i \geq \frac{1}{n} - \left(1-\sum_{u=1}^k (a_{iu}-a_{ju})\right) && i,j=1,\dots,n,\, i<j,\nonumber \\[-0.3cm] &&& k=1,\dots,n \label{eq:M6_q}\\
    & y_j-y_i \geq \frac{1}{n} - \left(1-\sum_{u=1}^k (a_{iu}-a_{ju})\right) && i,j=1,\dots,n,\, i<j,\nonumber \\[-0.3cm] &&& k=1,\dots,n \label{eq:M6_r}\\
    & x_i \leq f+\frac{i-1}{n} &&  i=2,\dots,n\\
    & x_i \geq \frac{i}{n}-f && i=2,\dots,n\\
    & y_i \leq f+\frac{i-1}{n} && i=2,\dots,n\\
    & y_i \geq \frac{i}{n}-f \label{eq:M6_v} && i=2,\dots,n
\end{align}
\end{subequations}
Constraint \eqref{eq:white} uses Theorem~\ref{th:white}. Constraints \eqref{eq:M6_o} to \eqref{eq:M6_v} are direct analogies to constraints \eqref{eq:2dcont_boundx2i} to \eqref{eq:4t}. Only constraints \eqref{eq:M6_q} and \eqref{eq:M6_r} have to be adapted since we no longer use ordering variables (the ordering is naturally defined by the grid itself).

\subsection{Experimental results}\label{sec:results}
We describe in this section the results obtained by our different models. All experiments were run with Gurobi 10.0.0 with an accuracy of $10^{-4}$ using Julia and the JuMP package. Experiments were run on a single machine of the MeSU cluster at Sorbonne Université, Intel Xeon CPU E5-2670 v3 with 24 cores. Tables~\ref{tab:lown} and~\ref{tab:cont_improved} give the optimal discrepancy values for point sets for $n=1$ to $n=21$ for different models, as well as the associated runtimes. M5 corresponds to the continuous formulation described in \eqref{eq:2dcont_dummy} and M6 corresponds to the assignment formulation \eqref{eq:M5_2dcont}. The ``+'' then indicates which extra constraints were added. The runtime is indicated for each formulation, the returned value is logically always the same. We note that values for $n=1$ to $n=6$ correspond to those known previously. The continuous formulation M5 is significantly faster than the assignment one M6, both with and without the extra constraints. We note that the extra constraints hinder M6 while their effectiveness for M5 is debatable: initial tests on a small laptop provided a factor 4 speedup, but there is no obvious improvement for the tests on the cluster shown below. We compare the discrepancies of our optimal sets to one of the most famous sets, obtained from the Kronecker sequence with golden ratio which we call Fibonacci set (see Section \ref{sec:latt} for a brief description).

\begin{table}[h]
\centering
\begin{tabular}{|l|r|r|r|r|r|r|r|r|r|}\hline
$n$ & $2$ & $3$ & $4$ & $5$ & $6$ & $7$ & $8$ & $9$  \\ \hline\hline
M5 & 0.01 & 0.02 & 0.63 & 0.41 &0.78&  0.81 & 0.3 & 0.93 \\ \hline
M5$_{+h\dots r}$ &0.01 & 0.1 & 0.01 & 0.14 &0.07& 0.2 & 0.32 & 0.7 \\ \hline
M6$_{+h\dots p}$  &0.05 &0.05 & 0.38&0.29&0.33 & 0.58 & 0.80& 7.19 \\ \hline
$d^*_{\infty}(P_{\text{cont}})$ 
& 0.3660
& 0.2847
& 0.2500
& 0.2000
& 0.1667
& 0.1500
& 0.1328
& 0.1235
\\ \hline
$d^*_{\infty}(Fib)$& 0.6909& 0.5880 & 0.4910& 0.3528& 0.3183&0.2728&0.2553&0.2270\\ \hline
\end{tabular}
\caption{$d=2$, $n\geq 2$ points located in the continuous box $[0,1]^2$. All problems solved with Gurobi to global optimality with a tolerance of $10^{-4}$.}
\label{tab:lown}
\end{table}
\begin{table}[h]
\centering
{
\begin{tabular}{|l|r|r|r|r|r|r|r|}\hline
$n$ & $10$ & $11$ & $12$ & $13$ & $14$ & $15$ \\ \hline\hline
M5 & 1.46 & 5.06 & 7.29 & 16.98 & 62.22 & 70 \\ \hline
M5$_{+h\dots r}$ & 0.91 &4.89 &12.01 &21.53 &69.94 & 110.56\\ \hline
M6$_{+h\dots p}$ & 5.47 & 15.41 & 27.61 & 65.31 & 6\,279.16 & 3\,445.26 \\ \hline
$d^*_{\infty}(P_{\text{cont}})$ & 0.1111 
& 0.1030
& 0.0952
& 0.0889
& 0.0837
& 0.0782
\\
\hline
$d^*_{\infty}(Fib)$& 0.2042 & 0.1857& 0.1702& 0.1571&0.1459& 0.1390\\ \hline
\hline
$n$ & $16$ & $17$ & $18$ & $19$ & $20$ & 21\\ \hline \hline
M5$_{+h\dots r}$ & 426.01 &616.45 &4\,610.23 &11\,240.5 &21\,892.76& 77\,988.0 
\\ \hline
M6$_{+h\dots p}$ & 12\,974.11& 22\,020.44& & & & \\ \hline
$d^*_{\infty}(P_{\text{cont}})$ & 0.0739
& 0.06996
& 0.0667
& 0.0634
& 0.0604
& 0.0580
\\
\hline
$d^*_{\infty}(Fib)$& 0.1486 & 0.1398& 0.1320& 0.1251& 0.1188& 0.1132\\ \hline
\end{tabular}
}
\caption{$d=2$, $n\geq 10$ points located in the continuous box $[0,1]^2$. 
While in almost all instances an optimal solution was obtained very quickly, it often took a very long time to close the duality gap and to prove optimality. All problems solved with Gurobi 10.0.0 to global optimality with a tolerance of $10^{-4}$.\label{tab:cont_improved}}
\end{table}

During the computational experiments we observed that very good -- if not optimal -- solutions were often found quite early during the solution process,
and a large amount of time was then used by the solver to close the duality gap, i.e., to prove global optimality of the solution. Thus, very good point sets can be obtained by running a solver with a pre-specified time limit, however without showing optimality. 
Table~\ref{tab:cont_heuristic} and Figure~\ref{fig:cont_heuristic} shows discrepancy values obtained for Model M5$_{+h}$ within a time limit of 1\,000s, 10\,000s, 20\,000s and 40\,000s, respectively. It turns out that for $n$ larger than 100, the complexity of the model increases too much to make this a viable approach. For the 100 point set, all the results before the 40\,000 seconds cutoff are worse than the Fibonacci set, only the last improvement makes it noticeably better. We note that Model M5$_{+h}$ outperformed Model M5$_{+h,\ldots,r}$ in nearly all cases: only some duality gaps are smaller with 1~000 seconds runtime.

\begin{table}[h]
\centering
\begin{tabular}{|l|r|r|r|r|r|r|r|}\hline
$n$ & 30 & 40 & 50 & 60 &80 & 100
\\ \hline\hline
\multicolumn{7}{|l|}{time limit 1\,000s:}\\ \hline
$d^*_{\infty}(P_{\text{cont}})$ &  0.04305 
& 0.035
& 0.0317
& 0.02934
& 0.02467
& 0.04070
\\ \hline
gap &  0.01718
& 0.01780 
& 0.02432 
& 0.02567
& 0.01807
& 0.03568
\\ \hline\hline
\multicolumn{7}{|l|}{time limit 10\,000s:}\\ \hline
$d^*_{\infty}(P_{\text{cont}})$ & 0.0426 
& 0.0334
& 0.0283
& 0.02469
& 0.02467
& 0.02570
\\ \hline
gap & 0.01296
& 0.01087 
& 0.01550
& 0.01682 
&  0.01781
& 0.02068
\\ \hline\hline
\multicolumn{7}{|l|}{time limit 20\,000s:}\\ \hline
$d^*_{\infty}(P_{\text{cont}})$ & 0.0426 & 0.0333 & 
0.0283 & 0.0246 & 0.02379 & 0.02346
\\ \hline
gap & 0.01208 & 0.01015 & 
0.01530 & 0.01682  &0.01653 & 0.01844
\\ \hline\hline
\multicolumn{7}{|l|}{time limit 40\,000s:}\\ \hline
$d^*_{\infty}(P_{\text{cont}})$ & 0.0424 & 0.0332 & 0.028 & 0.02435 &0.02131 &0.01933
\\ \hline
gap & 0.01124 & 0.00947 & 0.01519 & 0.01666 & 0.01403 &  0.01431  
\\ \hline \hline
\multicolumn{7}{|l|}{Reference}\\ \hline
Fibonacci & 0.079231 & 0.063836 & 0.053068 & 0.044223 & 0.033167 & 0.027485 
\\ \hline
\end{tabular}
\caption{$d=2$, points located in the continuous box $[0,1]^2$. 
All problems solved for Model M5$_{+h}$  with Gurobi with a time limit of 1\,000s, 10\,000s, 20\,000 and 40\,000 seconds, respectively. Gap represents the difference between the current best solution and a lower bound found by the solver.
\label{tab:cont_heuristic}}
\end{table}

\begin{figure}
    \centering
    \includegraphics[width=0.6\textwidth]{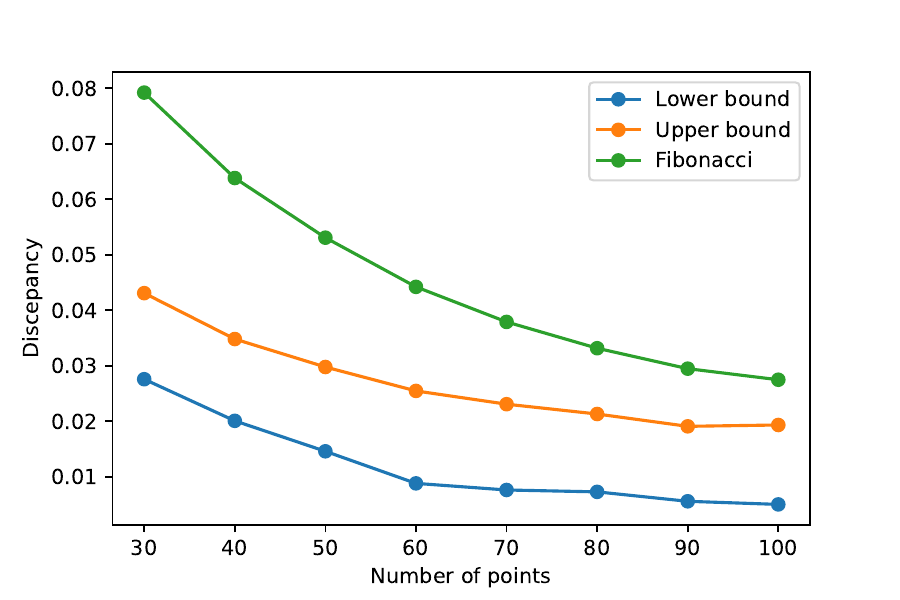}
    \caption{Comparison of the bounds at 40\,000s obtained in Table~\ref{tab:cont_heuristic} with the values of the Fibonacci set.}
    \label{fig:cont_heuristic}
\end{figure}

\subsection{Structural Differences between Known and Optimal Point Sets}\label{sec:images}

Finally, plotting the local discrepancy values over $[0,1]^2$ reveals that our optimal sets have a very different structure to known low-discrepancy sets. In each plot, we calculate the local discrepancy values over $[0,1)^2$, with brighter colors indicating a worse discrepancy. Note that each plot has its own color scale. Furthermore, ``triangles'' whose corner is to the top-right correspond to open boxes with too few points whereas those with a corner to the lower-left correspond to closed boxes with too many points. To facilitate a visual inspection, we also provide a \emph{truncated} version of these plots, where all local discrepancy values smaller than $d^{*}_{\infty}(P)-1/n$ are set to 0. This allows us to see better where the worst discrepancy values are reached and if the local discrepancy values are balanced over the whole space. 
While in sets like Fibonacci or the initial segments of the Sobol' sequence only a few closed boxes give active constraints for the discrepancy (i.e. the inequality is an equality), a much bigger set of boxes are very close to the maximal discrepancy value for our sets. In particular, the truncated plots show that a very large part of $[0,1)^2$ has a local discrepancy close to $1/n$ for our optimal sets. For both Fibonacci and Sobol' sets, only overfilled boxes appear, and this seems to be a characteristic independent of $n$. However, for our sets, both types of triangles appear and quite often sharing a box corner.

\begin{figure}
    \centering
    \includegraphics[width=0.8\textwidth]{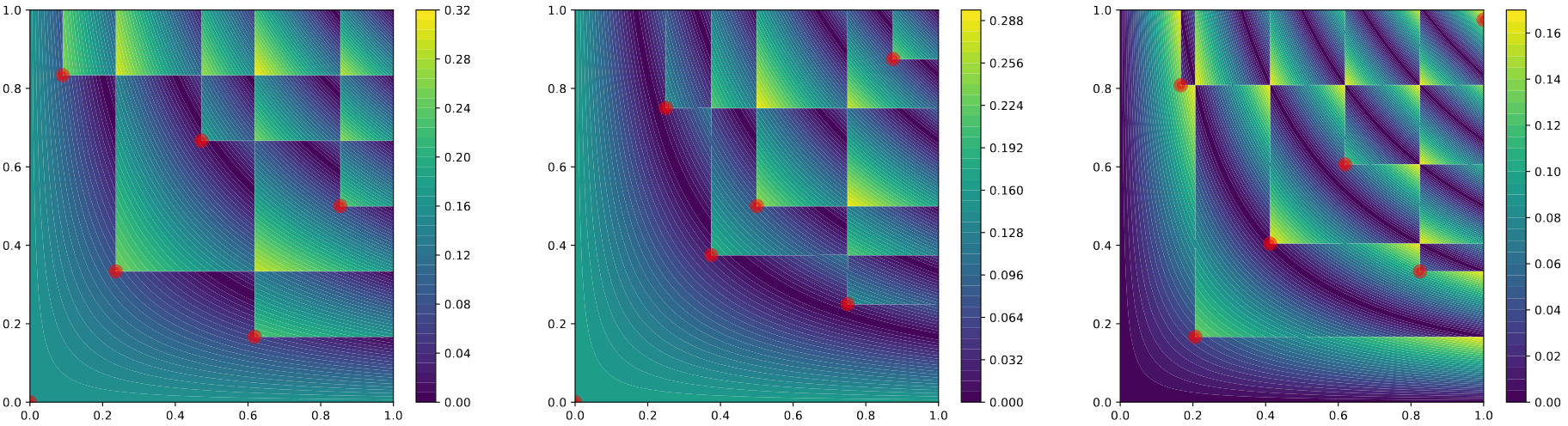}
    \caption{Fibonacci, Sobol' and optimal sets' local discrepancies for $n=6$.}
\end{figure}

\begin{figure}\centering
\includegraphics[width=0.8\textwidth]{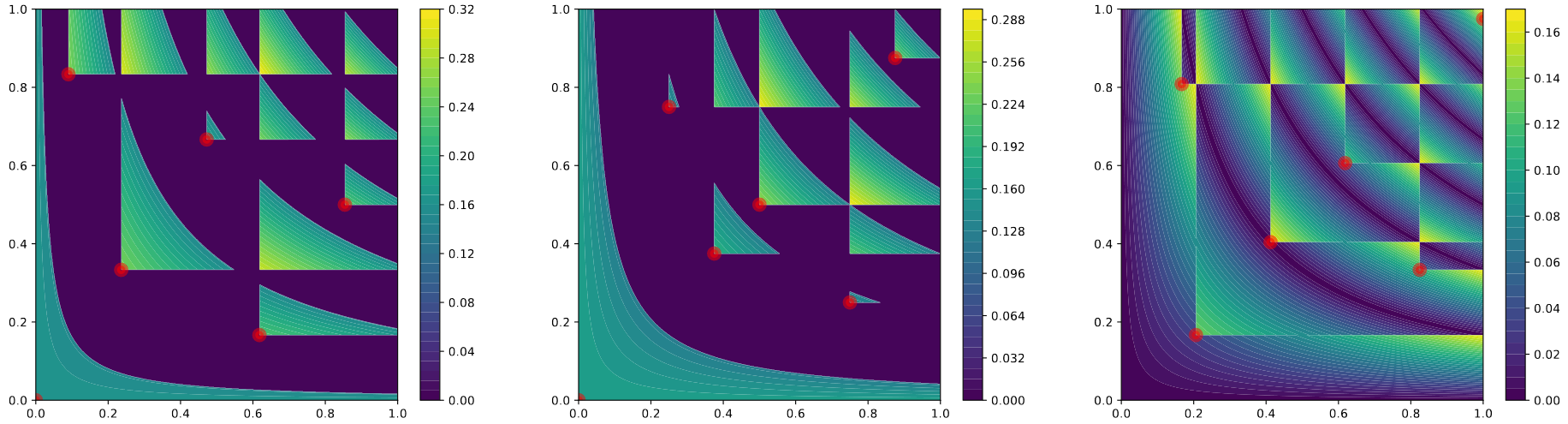}
    \caption{Fibonacci, Sobol' and optimal sets' \emph{truncated} local discrepancies for $n=6$. Local discrepancy values more than $1/n$ away from the star discrepancy were set to 0. Colored regions are therefore close to the worst discrepancy value.}
\end{figure}

\begin{figure}
    \centering
    \includegraphics[width=0.8\textwidth]{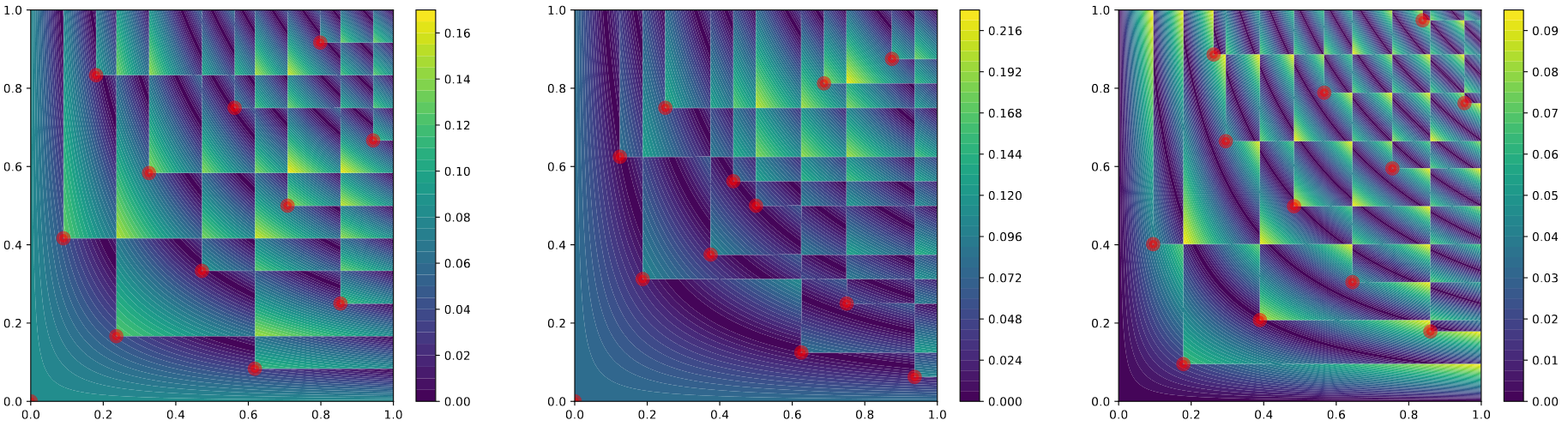}
    \caption{Fibonacci, Sobol' and optimal sets' local discrepancies for $n=12$.}
\end{figure}
    
\begin{figure}
    \centering    \includegraphics[width=0.8\textwidth]{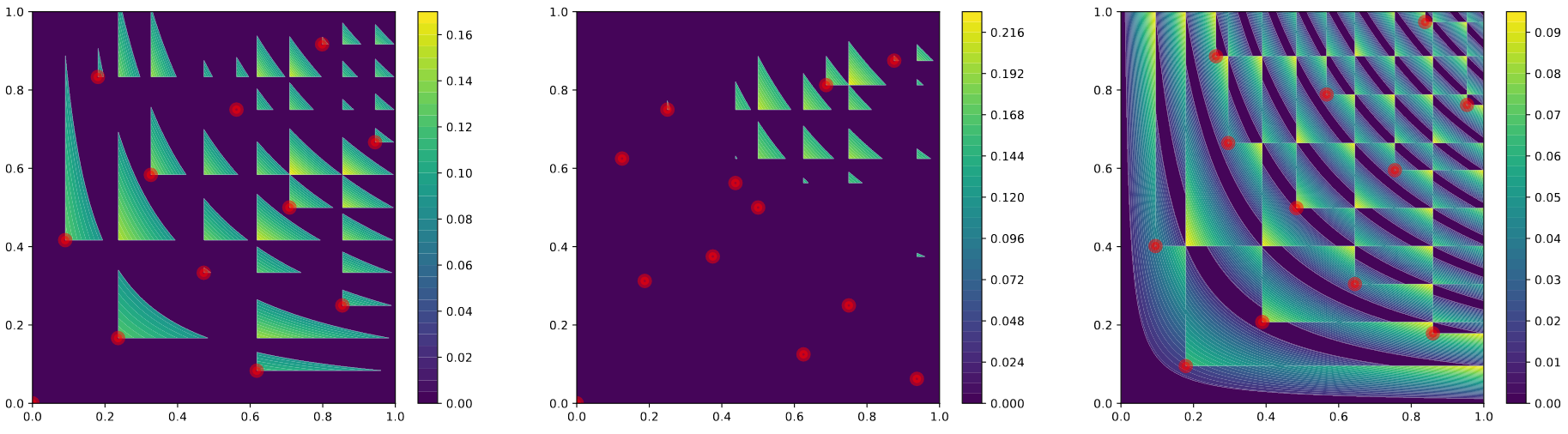}
    \caption{Fibonacci, Sobol' and optimal sets' truncated local discrepancies for $n=12$.}
\end{figure}

\begin{figure}
    \centering
    \includegraphics[width=0.8\textwidth]{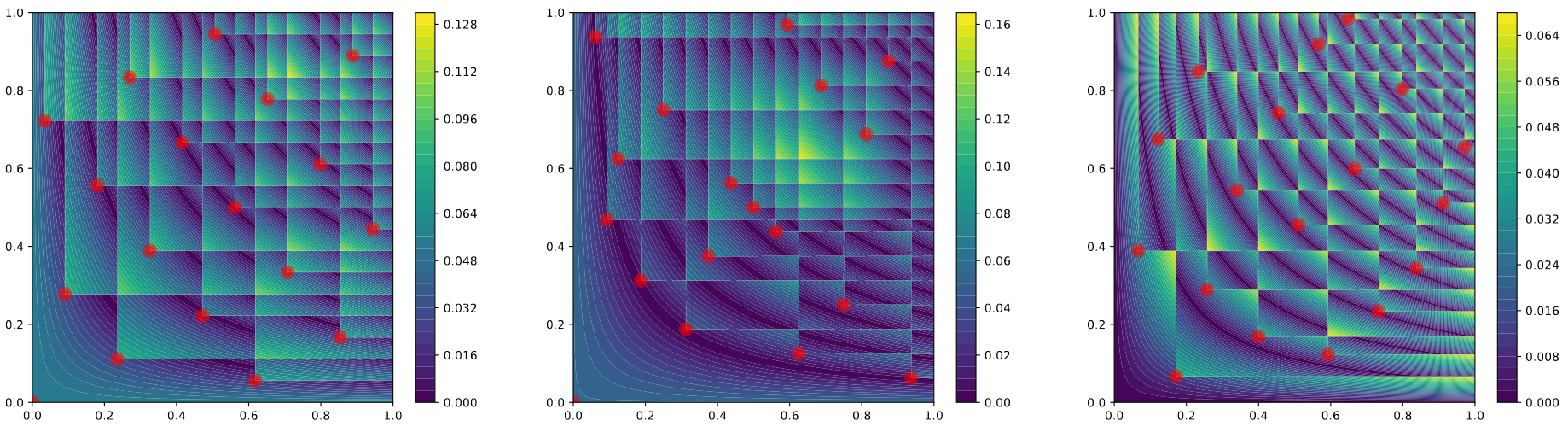}
    \caption{Fibonacci, Sobol' and optimal sets' local discrepancies for $n=18$.}
\end{figure}
    
\begin{figure}
    \centering    \includegraphics[width=0.8\textwidth]{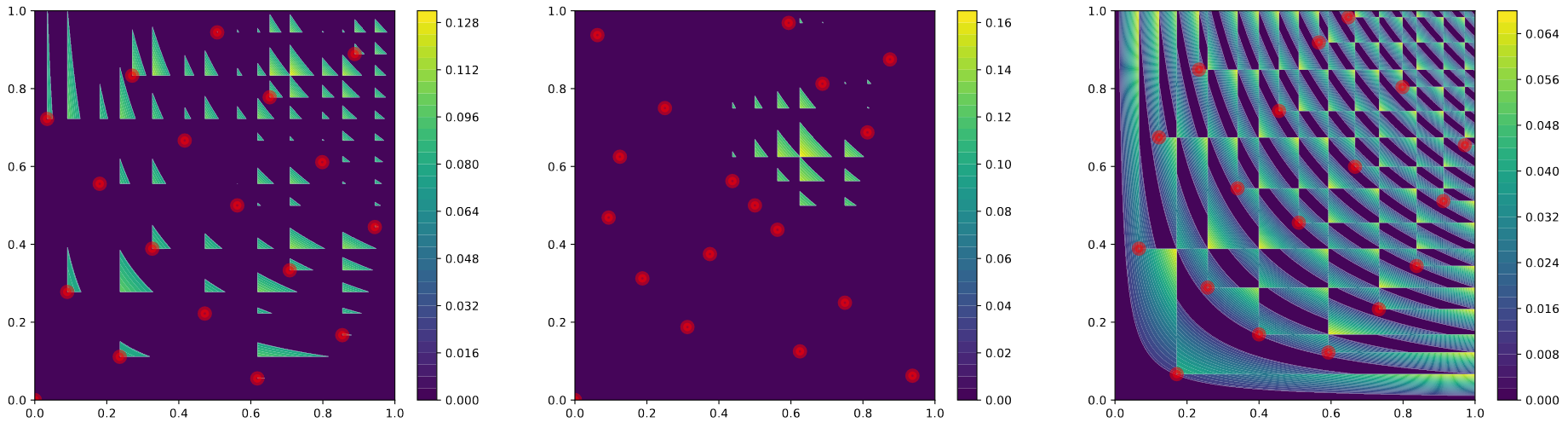}
    \caption{Fibonacci, Sobol' and optimal sets' truncated local discrepancies for $n=18$.}
\end{figure}

\section{Extensions}\label{sec:ext}

This section presents a number of possible extensions of our model. Section~\ref{sec:3d} introduces the most natural one, an extension to higher dimensions. Section~\ref{sec:latt} describes a way of obtaining optimal lattices, a very frequent structure in the discrepancy community. Finally, Section~\ref{sec:otherdisc} considers alternative discrepancy notions (extreme, periodic, and multiple-corner discrepancy for details). All these different measures appear frequently in the discrepancy literature \cite{Mat}, sometimes under the name ``axis-parallel'' boxes for the extreme discrepancy. For example, the paper by Hinrichs and Oettershagen \cite{Hinrichs2014} mentioned in the introduction tackled the optimal constructions for the periodic $L_2$ discrepancy (not star).
\subsection{An extension to three dimensions}\label{sec:3d}

Both models presented in Section~\ref{sec:Formu} naturally extend to higher dimensions. However, a considerably simpler model is obtained when generalizing model~\eqref{eq:M5_2dcont} to the three-dimensional case and will be the only one described here. Our experiments suggest it is also the fastest model to solve.
As in model~\eqref{eq:M5_2dcont}, we define the $x$-, $y$- and $z$-coordinates separately and independently from each other in vectors $x,y,z\in[0,1]^n$. The point set is then obtained by assigning exactly one $y$- and one $z$-coordinate to every $x$-coordinate, and vice versa, using assignment variables $a_{ijk}\in\{0,1\}$, $i,j,k=1,\dots,n$.

\begin{subequations}\label{eq:M5_3dcont}
\begin{align}
    \min\;\; & f && \nonumber\\ 
    \text{s.t.}\;\; & \displaystyle \frac{1}{n} \sum_{u=1}^i\sum_{v=1}^j\sum_{w=1}^k a_{uvw} - x_{i}y_{j}z_k \leq f && \forall i,j,k=1,\dots,n \label{eq:M5_upperboundcont_3}\\
    & \displaystyle \frac{-1}{n} \sum_{u=1}^{i-1}\sum_{v=1}^{j-1}\sum_{w=1}^{k-1} a_{uvw} + x_{i}y_{j}z_k \leq f && \forall i,j,k=1,\dots,n+1 \label{eq:M5_lowerboundcont_3}\\
    & x_{n+1}=1,\; y_{n+1}=1, \; z_{n+1}=1 &&\\
    & \sum_{j=1}^n\sum_{k=1}^n a_{ijk} = 1 && \forall i=1,\dots,n\\
    & \sum_{i=1}^n\sum_{k=1}^n a_{ijk} = 1 && \forall j=1,\dots,n\\
    & \sum_{i=1}^{n}\sum_{j=1}^n a_{ijk} = 1 && \forall k=1,\dots,n\\
    & x_i,y_i,z_i\in[0,1] && \forall i=1,\dots,n; \nonumber \\
    &a_{ijk}\in\{0,1\} && \forall i,j,k=1,\dots,n \nonumber \\
    &f\geq 0 \hspace*{-4cm} && \nonumber 
\end{align}
\end{subequations}

The values of $x_{n+1},y_{n+1},z_{n+1}$ are fixed to one and used as dummy values.

One should be careful that our result on the general position of some optimal sets no longer holds in higher dimensions. While the admissible shifts still do not increase the discrepancy of the set (Lemmas~\ref{lemma:upshift_ub} to~\ref{lemma:downshift_ub} still hold), a pair of points $x^{(1)}$ and $x^{(2)}$ sharing a coordinate no longer guarantees that one dominates the other. This means we can no longer use the constraints lower-bounding the gaps between different coordinates. Furthermore, some of the local discrepancy constraints will count the points incorrectly. Nevertheless, for every distinct box, there will be a constraint correctly computing the local discrepancy. Indeed, suppose there exist $x_k=x_i$ and $i<k$, two coordinates in the same dimension of different points. For the closed box defined by $x_i$, we will not count $x_k$, which leads to a lower discrepancy lower bound. However, the closed box for $x_k$ counts these points correctly: the lower bound is unchanged. A similar argument holds by swapping the roles for open boxes.

As before, the formulation can be strengthened by the following constraints, as the lowest point in each dimension can still be shifted:

\setcounter{equation}{6}
\begin{subequations}
\setcounter{equation}{9}
\begin{align}
    & x_1 = f &&\\
    & y_1 = f &&  \\
    & z_1 = f && \\
    & f \geq 1/n && \text{valid if $n \geq 3$}
\end{align}
\end{subequations}
This is based on the fact that Lemmas~\ref{lemma:upshift_ub}, \ref{lemma:upshift_lb}, \ref{lemma:downshift_lb} and \ref{lemma:downshift_ub} immediately generalize to three and more dimensions.

We note that extensions to higher dimensions could be obtained in a similar way. However, we would obtain $O(n^d)$ constraints just to bound the discrepancy, without even considering the products of many variables: each local discrepancy constraint has a product of $d$ variables.

Table~\ref{tab:3d_improved} describes our results in dimension 3. M7 corresponds to the 3$d$ formulation introduced above in \eqref{eq:M5_3dcont}. We are able to solve it for $n \leq 8$. 
While the model is relatively easy to solve with the general position hypothesis (518.98s for $n=8$), it is vastly more expensive without. Indeed, it took 1 569 728s to solve the problem in the $n=8$ case. Finding an optimal solution is relatively inexpensive (around a day), but closing the optimality gap takes a lot longer.

Table~\ref{tab:cont_heuristic} in the 2$d$ case suggests this is not an insurmountable obstacle if the objective is only to obtain excellent sets, and not optimal ones. Runtime values are not included as the computer was not running exclusively this experiment and may not be perfectly representative.
\begin{table}[h!]
\centering
\begin{tabular}{|l|r|r|r|r|r|r|r|r|}\hline
$n$ & $1$ & $2$ & $3$ & $4$ & $5$ & $6$ & $7$ & $8$   \\ \hline\hline
$d^*_{\infty}(P_{\text{cont}})$ & 0.6823 & 0.4239 & 0.3445 & 0.3038 & 0.2618 & 0.2326 & 0.2090 & 0.1875\\ \hline
\end{tabular}
\caption{$d=3$, $n\geq 1$ points located in the continuous box $[0,1]^3$. All problems solved with Gurobi 10.0.0, using a simple reformulation of the cubic terms into quadratic terms.}
\label{tab:3d_improved}
\end{table}

\subsection{Optimal lattice construction}\label{sec:latt}

A common construction for low-discrepancy sets is to build lattices $\{zi/n:z\in \mathbb{N}^d, i\in \{0,\ldots,n-1\}\}$. Choosing the parameter $z$ well is a difficult problem \cite[Section 5]{Nie92}. A typical choice is usually integers $z=(z_1,\ldots,z_d)$ with no factor in common with $n$. These sets are particularly important in practice as they allow to exploit increasing smoothness of a function when trying to numerically calculate an integral. They have been very extensively studied, both with integer parameters or more generally with elements of $\mathbb{R}^s/\mathbb{Z}^s$ (general lattice rules, which is closer to what we obtain here). In particular, choosing rationals $z_j/n$ with small continued fraction expansion leads to good lattices (we refer once again to the book by Niederreiter \cite{Nie92}). Setting aside integer parameters, this continued fraction expansion property is reflected by the excellent performance of the Kronecker sequence with golden ratio. It is defined by $(\{n\phi\})_{n\in \mathbb{N}}$, where $\phi$ is the golden ratio and $\{n\phi\}$ is the fractional part of $n\phi$. This can be extended to a 2-dimensional set of lattice form also called Fibonacci set: $\big\{(i/n,\{i\phi\}):i \in \{0,\ldots,n-1\}\big\}$. Construction methods for these lattices often involve either component by component search, exhaustive search or \emph{Korobov} constructions, where the lattice parameter is given by $(1,a,a^2,\ldots,a^{d-1})$ for some $a$. Our models suggest another construction method for these lattices.

We now want to build a point set $P_r=\big\{(i/n,\{ir\})|i \in \{0,\ldots,n-1\}\big\}$, with $r \in [0,1)$. While optimizing on integers $i$ coprime with $n$ is not so convenient, we can obtain even better sets by considering all reals in $[0,1]$. This can be easily implemented in the first continuous model, by adding constraints on the relations between two consecutive points in the set. We have that $x_{2i-1}=(i-1)/n$, $x_2=0$ and $x_{2i}+r=x_{2(i+1)} \mod 1$ for $ i \in \{1, \dots,n-1\}$. The first two equalities pose no problem and can be directly added as is in the model. For the third, we add binary variables $(k_i)_{i \in \{2,\ldots,n\}}$, where $k_i=1$ if $x_{2i}+r>1$. We then obtain the equality $x_{2i}+r=x_{2(i+1)}+k_i$, which can be added to our model. While this was implemented only for model \eqref{eq:2dcont_dummy}, this change can also be done for model \eqref{eq:M5_2dcont}. For model~\eqref{eq:M5_2dcont}, the only choices are then the value of $y_1$ and which $a_{1j}$ we pick, as they will define all the other variables. Equation \eqref{eq:mindist} can be removed since we know that the distance between two consecutive $x$-coordinates is $1/n$. The following constraints then need to be added to model \eqref{eq:2dcont_dummy}.

\begin{subequations}\label{eq:lattice_1}
\begin{align}
& x_{2i-1}=(i-1)/n && i=1,\ldots,n \label{eq:latt_xcoo}\\
& x_2=0 \\
& x_{2i}+r=x_{2i+2}+k_i && i=1,\ldots,n-1 \label{eq:latt_ycoo} \\
& k_i \geq x_{2i}+r-1+\varepsilon && i=1,\ldots,n-1 \label{eq:latt_biny}\\
& k_i \in \{0,1\} && i=1,\ldots,n-1 \nonumber\\
& r \in [0,1] \nonumber 
\end{align}
\end{subequations}

We note that a very small $\varepsilon$ needs to be added in constraint~\eqref{eq:latt_biny} because the strict inequality is not handled well by the solver. The other possible solution is to consider the inequality without the $\varepsilon$ and verify that there is no point with $1$ as a coordinate.

Table~\ref{tab:simple_lattice} gives the best discrepancy values obtainable for lattice sets with one parameter fixed to $1/n$. While these are worse than the optimal values for $n=20$, they are still better than the Fibonacci sequence by a decent margin.

\begin{table}[h!]
    \centering
    \begin{tabular}{|c|c|c|c|c|}
    \hline
      $n$ & $r$ & $d^*_{\infty}(LAT_1)$ & $d_{\infty}^*(Fib)$  & Runtime (s)\\ \hline
        20 & 0.653 & 0.094898& 0.1188  & 7.05 \\ \hline
        25 & 0.64269 & 0.077410& 0.095078  & 21.82 \\ \hline
        30 & 0.733 & 0.06492& 0.079231 & 149.85 \\ \hline
        35 & 0.82759 & 0.056137& 0.067913  & 797.66 \\ \hline
        40 & 0.79404 & 0.049594& 0.063836  & 2349.33 \\ \hline
    \end{tabular}
    \caption{Best discrepancies for lattice constructions in 2d for different $n$. $r$ corresponds to the lattice parameter.}
    \label{tab:simple_lattice}
\end{table}

This approach was then generalized to have lattices with two parameters: rather than $(1/n,\{ir\})$, we now consider $(\{ir_1\},\{ir_2\})$. These lattices will naturally be better than the single parameter ones. The constraints added to remove the fractional part now need to be used on the first variable as well. This requires removing equation \eqref{eq:latt_xcoo} and adding the following equations.

\setcounter{equation}{7}
\begin{subequations}\label{eq:lattice_2}
\begin{align}
\setcounter{equation}{3}
&x_1=0,x_2=0\\
& x_{2i-1}+r_2=x_{2i+1}+h_i && i=1,\ldots,n-1 \label{eq:latt_xcoo2} \\
& h_i>x_{2i-1}+r_2-1 && i=1,\ldots,n-1 \label{eq:latt_binx}\\
& h_i \in \{0,1\} && i=1,\ldots,n-1 \nonumber\\
& r_2 \in [0,1] \nonumber 
\end{align}
\end{subequations}

Table~\ref{tab:double_lattice} gives the two parameters as well as the best discrepancies obtained for varying $n$ by lattices. Interestingly, the first parameter $r_1$ is always very close to $1/n$, while still improving noticeably on the discrepancy values from Table~\ref{tab:simple_lattice}. It can also be seen that while the first parameters are monotonically decreasing and quite close to $1/n$, the second parameters seem to have similar values (a lot of them are around 0.725). We notice that the discrepancy of the best lattice with 37 points remains far from that of our optimal point set with 20 points. This suggests that construction methods other than lattices should be considered when trying to build high-quality low-discrepancy sets.

\begin{table}[h!]
    \centering
    \begin{tabular}{|c|c|c|c|c|c|c|}
    \hline
       $n$ & $r_1$ & $r_2$ & $d^*_{\infty}(LAT_2)$& $d^*_{\infty}(LAT_{1})$&$d^*_{\infty}(Fib)$ & Runtime (s) \\ \hline
        15 & 0.07 & 0.843 & 0.1083&  0.1233&0.1390 & 1.84 \\ \hline
        20 & 0.052 & 0.737 & 0.083795 &0.0949 &0.1188 & 11.21 \\ \hline
        25 & 0.0411 & 0.79167 & 0.0697 & 0.0774&0.095078 & 41.09 \\ \hline
        30 & 0.0343 & 0.784 & 0.05918 &0.0694 &0.079231 & 194.07 \\ \hline
        32 & 0.03218 & 0.72489 & 0.0613&0.055749 & 0.074279 & 731.00 \\ \hline
        34 & 0.03022 & 0.72497 & 0.0578&0.052556 & 0.069910 & 428.03 \\ \hline
        35 & 0.02936 & 0.72498 & 0.05107 &0.0562 &0.067913 & 1313.43 \\ \hline
        37 & 0.02777 & 0.72489 & 0.048303& 0.0534  & 0.067861 & 1295.44 \\ \hline
        40 & 0.0256 &0.725 &0.0449&0.0496 &0.063836 &3158.23 \\ \hline
    \end{tabular}
    \caption{Best discrepancies obtained for double lattices, as well as associated parameters. The Fibonacci and one-parameter lattice values are added as reference.}
    \label{tab:double_lattice}
\end{table}

\subsection{Other discrepancies}\label{sec:otherdisc}
We describe in this subsection three possible reformulations of our models to build optimal sets for other discrepancy measures commonly used in discrepancy theory. This is not an exhaustive list of possible extensions of our models. We favored constructions with the most common notions in the literature, or that could have an impact on practical applications of low-discrepancy sets.
\subsubsection{Extreme discrepancy}
The extreme discrepancy removes the constraint that the lower-left corner of the box needs to be in $(0,0)$, considering ``axis-parallel'' boxes rather than ``anchored'' ones. More formally, for a point set $P$, 
$$D_{\infty}^{ext}(P):=\sup_{q,r \in [0,1)^d, q < r} \left|\frac{|P \cap [q,r)|}{|P|} -\lambda([q,r))\right|,$$ 
where $q<r$ indicates that $q$ is smaller than $r$ for each coordinate. This discrepancy measure is more general than the $L_{\infty}$ star discrepancy, but given that it is even more complicated to compute, it has been studied less. It is mentioned in~\cite{Nie92,DGWBook}, but there are only very few papers focusing explicitly on it such as~\cite{Moro}. Our models can be easily adapted for this change of discrepancy measure: it requires changing the lower bound of each sum in the open and closed box constraints \eqref{eq:upperboundcont_dummy} and \eqref{eq:lowerboundcont_dummy} (model~\eqref{eq:2dcont_dummy}) or~\eqref{eq:M5_upperboundcont} and~\eqref{eq:M5_lowerboundcont} (model~\eqref{eq:M5_2dcont}), and considering the corresponding constraints for all relevant combinations of $q$ and $r$. For example, in model~\eqref{eq:M5_2dcont}, rather than always starting at 1 and going to $i$ in each sum, we start at a certain $1 < k \leq i$. We have a strict inequality here as we suppose in our models that there are no two points sharing a coordinate. While we do not provide a formal proof of this, and it is more a practical necessity, the results from Section~\ref{sec:mingap} should transfer also to this case, with adapted definitions of admissible shits. Some care should also be put into checking that the box defined by $(k,l)$ and $(i,j)$ is valid, in other words that $k$ is on the lower horizontal edge, $i$ on the upper horizontal edge, $l$ on the left vertical edge and $j$ on the right vertical edge. In model~\eqref{eq:2dcont_dummy} this can be checked using the indicator variables. In model~\eqref{eq:M5_2dcont}, this is trivially done since we rely on the ordering of the $x_i$'s and $y_j$'s and do not check for critical boxes.

Adapting model \eqref{eq:M5_2dcont} requires removing constraints \eqref{eq:M5_upperboundcont} and \eqref{eq:M5_lowerboundcont} and replacing them by

\begin{subequations}\label{eq:Ext_M5}
\begin{align}
& \displaystyle \frac{1}{n} \sum_{u=k}^i\sum_{v=l}^j a_{uv} - (x_{i}-x_{k})(y_{j}-y_{k}) \leq f && \forall i,j,k,l=0,\dots,n, k< i, l < j \label{eq:Ext_M5upp} \\
& \displaystyle \frac{-1}{n} \sum_{u=k+1}^{i-1}\sum_{v=l+1}^{j-1} a_{uv} + (x_{i}-x_{k})(y_{j}-y_{k}) \leq f && \forall i,j,k,l=0,\dots,n+1, k<i, l< j \label{eq:Ext_M5low} \\
& x_0=y_0=0,\; a_{0j}=a_{i0}=0 && \forall i,j=0,\dots,n. \label{eq:Ext_M5dummy}
\end{align}
\end{subequations}

These are the only constraints that have an impact on the discrepancy value, constraints \eqref{eq:M5_dummy} to \eqref{eq:M5_row} are unchanged. Given that the set of boxes defining the $L_{\infty}$ star discrepancy is a subset of those defining the extreme discrepancy, our lower bounds for the star discrepancy are also valid for the extreme discrepancy: constraint \eqref{eq:white} can be kept. However, constraints \eqref{eq:M6_o} to \eqref{eq:M6_v} proceeding from our results in Section~\ref{sec:mingap} do not directly generalize to this case. We note that we still use the $\varepsilon$ spacing with a small constant.

For the extreme discrepancy, the model has many more constraints. It is therefore expected that we are not able to obtain solutions for $n$ as high as in the $L_{\infty}$ star discrepancy case. Table~\ref{tab:ext_disc} describes the runtimes and discrepancy values obtained.

\begin{table}[h!]
    \centering
    \begin{tabular}{|c|c|c|c|c|c|c|c|c|c|c|}
    \hline
        $n$ & 1& 2 & 3& 4& 5& 6& 7& 8& 9& 10 \\
        \hline
        $D_{\infty}^{ext}(OPT_{ext})$ & 1& 0.618 & 0.467& 0.375& 0.312& 0.2697& 0.2489& 0.2292& 0.2143& 0.2   \\
        \hline
        Runtime (s) & 0.44& 0.22 & 0.16& 0.19& 1.93& 3.16& 26.78& 47.85& 129.87& 1313.29   \\
        \hline
        
    \end{tabular}
    \caption{Optimal discrepancies and runtimes for the $L_{\infty}$ extreme discrepancy with an adapted Model~\eqref{eq:M5_2dcont} }
    \label{tab:ext_disc}
\end{table}
\subsubsection{Periodic discrepancies}

Another common aspect in discrepancy measures is to consider $[0,1)^d$ as a $d$-dimensional torus (here $d=2$). The discrepancy is now defined by

$$ D_{\infty}^{per}(P):=\sup_{q,r \in [0,1)^2}\left|\frac{|P \cap [q,r)|}{|P|}-\lambda([q,r))\right|,$$
where each one-dimensional element composing $[q,r)$ is given by $[q_i,r_i)$ if $q_i \leq r_i$ and $[q_i,1) \cup [0,r_i)$ otherwise. This notion is also known as ``Weyl'', ``toroidal'', or ``wrap around'' discrepancy and was studied by Lev in \cite{Lev}. We notice that this is a further generalization of the set of boxes defined for the extreme discrepancy. It is notably more difficult than the star discrepancy. Indeed, as was recently shown by Gilibert, Lachmann and M\"{u}llner in \cite{GilLachMul}, the VC-dimension of the class of boxes required to define it is in $O(d\log(d))$ and not simply $O(d)$. Given a quadruplet $(i,j,k,l)$ where $i> k$ and $j> l$, a box can be of one of the four following types:
\begin{itemize}
    \item The regular extreme discrepancy box $[x_k,x_i) \times [y_l,y_j)$. The constraints are the same as \eqref{eq:Ext_M5upp}, \eqref{eq:Ext_M5low} and \eqref{eq:Ext_M5dummy}. 
    \item The box wrapping around the $x$-axis $([0,x_k)\cup[x_i,1))\times [y_l,y_j)$. The corresponding constraints are \eqref{eq:Per_C1U} and \eqref{eq:Per_C1L}. 
    \item The box wrapping around the $y$-axis $[x_k,x_i] \times ([0,y_l) \cup [y_j,1))$. This is associated with constraints \eqref{eq:Per_C2U} and \eqref{eq:Per_C2L}. 
    \item The box wrapping around both axes $([0,x_k)\cup[x_i,1))\times ([0,y_l) \cup [y_j,1))$. This corresponds to the two final constraints \eqref{eq:Per_C3U} and \eqref{eq:Per_C3L}.
\end{itemize}

Since every box defining the discrepancy has to share its axes with the $(x_i)$ and $(y_j)$ we define, one can easily check that all candidate boxes are covered by one of the types above for a specific quadruplet $(i,j,k,l)$ where $i>k$ and $j>l$.

By the arguments described above, the model for the periodic discrepancy can be obtained by taking that of the extreme discrepancy to which we add the following constraints. Once again, we are modifying model \eqref{eq:M5_2dcont}, but it is also possible to adapt model \eqref{eq:2dcont_dummy}. We take the same constraints as for the extreme case, including the dummy variables, and add the following constraints.

\begin{subequations}\label{eq:Per_M5}
\begin{align}
& \displaystyle \frac{1}{n} \sum_{v=l}^j (\sum_{u=1}^{k} a_{uv} \!+\!  \sum_{u=i}^{n} a_{uv}) \!-\!  (1\!-\! x_{i}\!+\! x_{k})(y_{j}\!-\! y_{l}) \leq f && \forall i,j,k\!=\! 1,\dots,n,\,l\!=\! 0,\ldots,n,\, k \!<\!  i, l \!<\!  j \label{eq:Per_C1U} \\
& \displaystyle \frac{-1}{n} \sum_{v=l+1}^{j-1}( \sum_{u=1}^{k-1} a_{uv} \!+\!\!  \sum_{u=i+1}^n \!\!a_{uv}) \!+\!  (1\!-\! x_{i}\!+\! x_{k})(y_{j}\!-\! y_{l}) \leq f && \forall i,j,k\!=\! 1,\dots,n\!+\! 1,\,l\!=\! 0,\ldots,n,\, k\!<\!  i, l\!<\!  j  \label{eq:Per_C1L}\\
& \displaystyle \frac{1}{n} \sum_{u=k}^i (\sum_{v=1}^{l} a_{uv} \!+\!  \sum_{u=j}^{n} a_{uv}) \!-\!  (x_{i}\!-\! x_{k})(1\!-\! y_{j}\!+\! y_{l}) \leq f && \forall i,j,l\!=\! 1,\dots,n,\,k\!=\! 0,\ldots,n,\, k \!<\!  i, l \!<\!  j \label{eq:Per_C2U} \\
& \displaystyle \frac{-1}{n} \sum_{u=k+1}^{i-1}( \sum_{v=1}^{l-1} a_{uv} \!+\!\!  \sum_{v=j+1}^n \!\! a_{uv}) \!+\!  (x_{i}\!-\! x_{k})(1\!-\! y_{j}\!+\! y_{l}) \leq f && \forall i,j,l\!=\! 1,\dots,n\!+\! 1,\,k\!=\! 0,\ldots,n,\, k \!<\!  i, l \!<\!  j  \label{eq:Per_C2L}\\
& \displaystyle \frac{1}{n}( \sum_{u=1}^k \sum_{v=1}^{l} a_{uv}\!+\!  \sum_{u=i}^n \sum_{u=j}^{n} a_{uv})  \!-\!  (1\!-\! x_i)(1\!-\! y_j)\!-\! x_ky_l \leq f && \forall i,j,k,l\!=\! 1,\dots,n, k\!<\!  i, l \!<\!  j \label{eq:Per_C3U} \\
& \displaystyle \frac{-1}{n} (\sum_{u=1}^{k-1}\sum_{v=1}^{l-1} a_{uv} \!+\!\! \sum_{u=i+1}^{n}\sum_{v=j+1}^n \!\! a_{uv}) \!+\!  (1\!-\! x_i)(1\!-\! y_j)\!+\! x_ky_l \leq f && \forall i,j,k,l\!=\! 1,\dots,n\!+\! 1, k\!<\!  i, l\!<\!  j  \label{eq:Per_C3L}
\end{align}
\end{subequations}

Table~\ref{tab:per_disc} describes the run times and discrepancy values obtained. The periodic discrepancy model contains more constraints than the extreme one, and hence the maximal $n$ for which we are able to find optimal point sets in reasonable time is further reduced.
\begin{table}[h!]
    \centering
    \begin{tabular}{|c|c|c|c|c|c|c|c|}
    \hline
        $n$ & 1& 2 & 3& 4& 5& 6& 7 \\
        \hline
        $D_{\infty}^{per}(OPT)$ & 1 & 0.75 & 0.618 & 0.5 & 0.4249 & 0.3671 & 0.3279    \\
        \hline
        Runtime & 0.1 & 0.6 &0.63 & 0.84& 53.80& 98.66 & 597.78    \\
        \hline
        
    \end{tabular}
    \caption{Optimal discrepancies and runtimes for the $L_{\infty}$ periodic discrepancy.}
    \label{tab:per_disc}
\end{table}

\subsubsection{Multiple-corner discrepancy}

Unlike the two previous notions which are common in theoretical discrepancy literature, this notion targets applications.
The $L_{\infty}$ star discrepancy creates an asymmetry by giving more importance to $(0,\ldots,0)$. Often, it is undesirable to orient the regularity of our point set based on a specific corner of the space. To limit the impact of this, a possible counter-measure is to consider the $L_{\infty}$ star discrepancy according to all $2^d$ corners, the worst of which is then the discrepancy value. In dimension 2, this \emph{4-corner discrepancy} is defined by the following
$$D_{\infty}^{4cor}(P):=\sup_{q=(q_1,q_2) \in [0,1)^2} \max\{ D([0,q),P),D((q,1],P),D([0,q_1)\times (q_2,1],P),D((q_1,1]\times [0,q_2),P)\}, $$
where $D(q,P)$ is the local discrepancy of the box $q$ for the point set $P$, c.f.\ Section~\ref{sec:Gene}.

For a given point set $P$, this can be reformulated w.r.t.\ the classical $L_{\infty}$ star discrepancy by considering, in addition, the sets $P_2:=\{(1-x,y):(x,y)\in P\}$, $P_3:=\{(1-x,1-y):(x,y)\in P\}$ and $P_4:=\{(x,1-y):(x,y) \in P\}$.

We then have 
$$D_{\infty}^{4cor}(P)= \max(d_{\infty}^{*}(P),d_{\infty}^{*}(P_2),d_{\infty}^{*}(P_3),d_{\infty}^{*}(P_4)).$$
Given the definition above, the definitions of the point sets, and our models for the star discrepancy, it can be easily seen that we simply need to make four copies of the constraints \eqref{eq:M5_upperboundcont} and \eqref{eq:M5_lowerboundcont} to adapt model \eqref{eq:M5_2dcont} to this situation. In addition to the constraints for $P$ given previously, the following constraints need to be added.

\begin{subequations}\label{eq:4c}
\begin{align}
& \displaystyle \frac{1}{n} \sum_{u=i}^n \sum_{v=1}^j a_{uv} - (1-x_i)y_j \leq f && \forall i,j=1,\dots,n \label{eq:4c_P1U} \\
& \displaystyle \frac{-1}{n} \sum_{u=i+1}^{n} \sum_{v=1}^{j-1} a_{uv} + (1-x_i)y_j \leq f && \forall i,j=1,\dots,n+1 \label{eq:4c_P1L}\\
& \displaystyle \frac{1}{n} \sum_{u=i}^n\sum_{v=j}^n a_{uv} - (1-x_{i})(1-y_{j}) \leq f && \forall i,j=1,\dots,n \label{eq:4c_P2U} \\
& \displaystyle \frac{-1}{n} \sum_{u=j+1}^{n}\sum_{v=i+1}^{n} a_{uv} + (1-x_{i})(1-y_{j}) \leq f && \forall i,j=1,\dots,n+1 \label{eq:4c_P2L}\\
& \displaystyle \frac{1}{n} \sum_{u=1}^i \sum_{v=j}^n a_{uv} - x_{i}(1-y_{j}) \leq f && \forall i,j=1,\dots,n \label{eq:4c_P3U} \\
& \displaystyle \frac{-1}{m} \sum_{u=1}^{i-1}\sum_{v=j+1}^{n} a_{uv} + x_{i}(1-y_{j}) \leq f && \forall i,j=1,\dots,n+1 \label{eq:4c_P3L}
\end{align}
\end{subequations}

 Finally, Table~\ref{tab:4c} gives the results obtained by our model with the multiple-corner discrepancy. We specify both the multiple-corner discrepancy of our set, and its actual $L_{\infty}$ star discrepancy (necessarily smaller or equal as all the $L_{\infty}$ star discrepancy constraints are included in the model). We note that the $L_{\infty}$ star discrepancy values are not much higher than the optimal ones, and far better than previously known sets. This suggests that considering a set with good multiple-corner discrepancy could be a good tradeoff for applications.

\begin{table}[h]
    \centering
    \begin{tabular}{|c|c|c|c|c|c|c|c|c|c|}
    \hline
        $n$ & 1 & 2 & 3 & 4 & 5 & 6 & 7 & 8 & 9 \\
        \hline
        $D^*_{\infty}(OPT_{star})$ & $\frac{1+\sqrt{5}}{2}$ & 0.366 & 0.2847 &0.2500 & 0.2000 & 0.1667 & 0.1500 & 0.1328 & 0.1235\\
        \hline
        $D^*_{\infty}(Fib)$ & 1.0 & 0.6909 & 0.5880 & 0.4910 & 0.3528 & 0.3183 & 0.2728 & 0.2553 & 0.2270 \\
        \hline
        $D_{\infty}^{*}(OPT_{4c})$ & 0.75 & 0.5 & 0.3333& 0.2548 & 0.2166& 0.1875 &0.1632 & 0.1438 & 0.1319 \\
        \hline
        $D_{\infty}^{4cor}(OPT_{4c})$ & 0.75 & 0.5& 0.3333 & 0.2548& 0.2166& 0.1875 & 0.1632& 0.1438& 0.1319  \\
        \hline
        Runtime & 0.0 & 0.01 & 0.02 & 0.12 &0.3 & 0.68 &1.14 &1.97 &4.34\\
        \hline
        \multicolumn{10}{c}{ }\\
        
        \hline
        $n$ &10& 11 & 12 & 13 & 14 & 15 & 16 & 17 & 18 \\
        \hline
        $D_{\infty}^*(OPT_{star})$& 0.1111 & 0.1030 & 0.0952 & 0.0889 & 0.0837 & 0.0782 & 0.0739 & 0.06996 & 0.06667\\
        \hline
        $D^*_{\infty}(Fib)$ & 0.2042&0.1857 & 0.1702 & 0.1571 & 0.1459 & 0.1390 & 0.1486 & 0.1398 & 0.1320 \\
        \hline
        $D_{\infty}^{*}(OPT_{4c})$& 0.1197 &0.1083 & 0.09999& 0.09502  & 0.08874 & 0.08429  & 0.07916 & 0.07510 &0.0715\\
        \hline
        $D_{\infty}^{4cor}(OPT_{4c})$ & 0.1197&0.1083 & 0.09999& 0.09502  & 0.08874 & 0.08429  & 0.07916 & 0.07510 & 0.0715   \\
        \hline
        Runtime& 12.68 & 26.45 & 46.46 & 88.46 & 219.22 &783.74&1616.65 &3287.64 &7412.73\\
        \hline
    \end{tabular}
    \caption{Comparison of previously best values for low-discrepancy sets and our optimal sets, both the 4-corner optimal set $OPT_{4c}$ and the $L_{\infty}$ star discrepancy optimal set $OPT_{star}$.}
    \label{tab:4c}
\end{table}

\begin{figure}[h]
    \centering
    \includegraphics[width=0.9\textwidth]{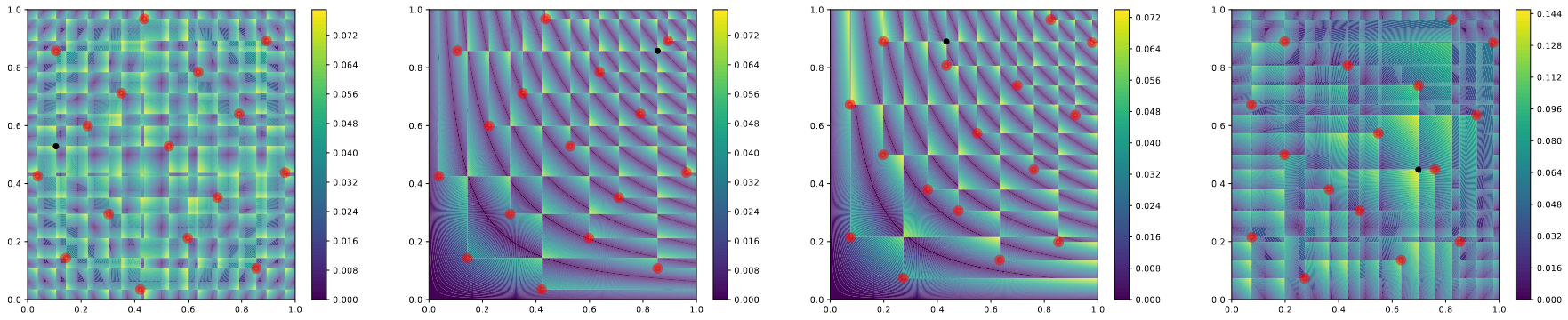}
    \caption{n=16. Left to right: local 4-corner discrepancy of the optimal 4-corner set, local $L_{\infty}$ star discrepancy of the optimal 4-corner set, local star discrepancy of the optimal $L_{\infty}$ star set and local 4-corner discrepancies of the optimal $L_{\infty}$ star set. While the optimal 4-corner set is very good for the $L_{\infty}$ star discrepancy, the same can't be said for the 4-corner discrepancy of the optimal $L_{\infty}$ star set.}
\end{figure}

\section{Conclusion}
Using two non-linear programming models, we provide optimal constructions for sets of minimal $L_{\infty}$ star discrepancy for $n=1$ to $n=21$ in two dimensions and $n=1$ to $n=8$ in three dimensions. For $n \geq 7$ in dimension 2, these sets have much lower discrepancy than the previously known sets. Analyzing the local discrepancies of these sets and comparing to known low-discrepancy sequences also shows a clear structural difference, suggesting that known sequences are over-sampling in the lower zones of $[0,1]^2$ compared to an optimal distribution. These optimal constructions therefore suggest that new approaches are possible to design low-discrepancy sets and sequences. We furthermore observed an almost negligible extra cost for point sets that are well-distributed not only for the origin $(0,\ldots,0)$ but also other corners of the square.

There are still numerous open questions to tackle, both for our exact approaches, and for the development of heuristics we have not mentioned here. Our models could be improved by finding new cuts, or clever branching methods to reduce the problem into a series of smaller subproblems. On the heuristics side, the field is wide open. For example, our models can be easily adapted to adding a point optimally to an existing set, with far fewer variables that start unassigned. This could allow the construction of excellent point sets for greater values of $n$. Simply constructing a very good point set in multiple steps has not been studied yet and could lead to promising results. On the theoretical side, the relation between the multiple corner discrepancy and the star discrepancy has not been studied properly yet. Proving that there is an optimal set in general position in dimensions higher than 3 also remains open, as well as our conjecture on a $d/2n$ lower-bound for $n$ sufficiently large.

More generally, mathematical programming methods are not widespread in the discrepancy community. We hope that this new approach and its flexibility to tackle many different problems linked to discrepancy will motivate further applications.

\section*{Acknowledgments}
We thank Art Owen for introducing us to the multiple-corner discrepancy and Michael Gnewuch for providing helpful references for the periodic discrepancy. This work was granted access to the HPC resources of the SACADO MeSU platform at Sorbonne Université.

This project was partially supported by PHC Procope 50969UF project funded by the German Academic Exchange Service (DAAD, Project-ID 57705092), the French Ministry of Foreign Affairs, and the French Ministry of Higher Education and Research.
This work was also partially supported by the “PHC PESSOA” program (project DISCREPANCY -- Discrepancy Problems - Algorithms and Complexity, number 49173PH), funded by the French Ministry for Europe and Foreign Affairs, the French Ministry for Higher Education and Research, and the Portuguese Foundation for Science and Technology (FCT). This work is partially funded by the FCT, I.P./MCTES through national funds (PIDDAC), within the scope of CISUC R\&D Unit - UID/CEC/00326/2020. 

\bibliographystyle{alpha}
\bibliography{sample}

\end{document}